\documentclass{lmcs}
\pdfoutput=1

\usepackage{lastpage}
\lmcsdoi{19}{2}{11}
\lmcsheading{}{\pageref{LastPage}}{}{}%
{Nov.~10,~2021}{May~19,~2023}{}

\usepackage[utf8]{inputenc}
\usepackage{amsmath}
\usepackage{amssymb}

\usepackage{tikz}

\usepackage{hyperref}
\usepackage[capitalise]{cleveref}
\usepackage{cmll}

\usepackage{nicefrac}

\usepackage{xstring}

\newcommand{\orange}[1]{{\color{orange}#1}}
\newcommand{\purple}[1]{{\color{purple}#1}}

\tikzstyle{vertex}=[circle,minimum size=10pt,inner sep=0pt,text height=1.5ex,text depth=.25ex]
\tikzstyle{g} = [draw,thick,dotted,-,green]
\tikzstyle{r} = [draw,thick,-,red]
\tikzstyle{b} = [draw,thick,-,black]
\tikzstyle{u} = [draw,thin,-,black,opacity=0]

\newcommand{\ThreeGraph}[4]{
{\begin{tikzpicture}[baseline=(v1.base)]
		\node[vertex] (v1) at (0, 0) {$\StrBetween[1,2]{,#1,}{,}{,}$};
		\node[vertex]  (v2) at (1, 0.5) {$\StrBetween[2,3]{,#1,}{,}{,}$};
		\node[vertex] (v3) at (1 ,-0.5 ) {$\StrBetween[3,4]{,#1,}{,}{,}$};
		\draw[#2]  (v1) -- (v2) ;
		\draw[#3] (v1) -- (v3) ;
		\draw[#4] (v2) -- (v3) ;
\end{tikzpicture} }}

\newcommand{\redge}[3][]{
  {
    \begin{tikzpicture}[baseline=(v1.base)]
    \node[vertex] (v1) at (0, 0) {$#2$};
    \node[vertex]  (v2) at (1, 0) {$#3$};
    \draw[r]  (v1) -- (v2) node [midway,above=-1pt] {$\scriptstyle {\smash{#1}}$};
  \end{tikzpicture}
  }
}

\newcommand{\gedge}[3][]{
  {
    \begin{tikzpicture}[baseline=(v1.base)]
    \node[vertex] (v1) at (0, 0) {$#2$};
    \node[vertex]  (v2) at (1, 0) {$#3$};
    \draw[g]  (v1) -- (v2) node [midway,above=-1pt] {$\scriptstyle {\smash{#1}}$};
  \end{tikzpicture}
  }
}

\newcommand{\FourGraph}[7]{
\raisebox{-0.4\height}{\begin{tikzpicture}
		\node[vertex]  (v1) at (0, 1) {$\StrBetween[1,2]{,#1,}{,}{,}$};
		\node[vertex]  (v2) at (1, 1) {$\StrBetween[2,3]{,#1,}{,}{,}$};
		\node[vertex] (v3) at (0, 0) {$\StrBetween[3,4]{,#1,}{,}{,}$};
		\node[vertex]  (v4) at (1, 0) {$\StrBetween[4,5]{,#1,}{,}{,}$};
		\draw[#2]  (v1) -- (v2) ;
		\draw[#3]  (v1) -- (v3);
		\draw[#4] (v1) -- (v4);
		\draw[#5]  (v2) -- (v3);
		\draw[#6] (v2) -- (v4);
		\draw[#7]  (v3) -- (v4);
\end{tikzpicture}} }

\newcommand{\FiveGraphCircCont}[2]{
\raisebox{-0.45\height}{\begin{tikzpicture}
		\node[vertex] (v1) at (-0.588,-0.809) {$\StrBetween[1,2]{,\fir,}{,}{,}$};
		\node[vertex] (v2) at (-0.951,0.309) {$\StrBetween[2,3]{,\fir,}{,}{,}$};
		\node[vertex] (v3) at (0,1) {$\StrBetween[3,4]{,\fir,}{,}{,}$};
		\node[vertex] (v4) at (0.951,0.309) {$\StrBetween[4,5]{,\fir,}{,}{,}$};
		\node[vertex] (v5) at (0.588,-0.809) {$\StrBetween[5,6]{,\fir,}{,}{,}$};
		\draw[\sec, bend angle=15, bend left] (v1) edge (v2) ;
		\draw[\thi] (v1) -- (v3);
		\draw[\fou] (v1) -- (v4);
                \draw[\fif, bend angle=15, bend right] (v1) edge (v5) ;
		\draw[\six, bend angle=15, bend left] (v2) edge (v3);
		\draw[\sev] (v2) -- (v4);
                \draw[\eig] (v2) edge (v5) ;
		\draw[\nin, bend angle=15, bend left] (v3) edge (v4);
		\draw[#1] (v3) edge (v5) ;
		\draw[#2, bend angle=15, bend left] (v4) edge (v5) ;
\end{tikzpicture} }}
\newcommand{\FiveGraphCirc}[9]{
	\def\fir{#1}
	\def\sec{#2}
	\def\thi{#3}
	\def\fou{#4}
	\def\fif{#5}
	\def\six{#6}
	\def\sev{#7}
	\def\eig{#8}
	\def\nin{#9}
	\FiveGraphCircCont
}

\newcommand{\FiveGraphRCont}[2]{
\raisebox{-0.45\height}{\begin{tikzpicture}
		\node[vertex] (v1) at (0,1) {$\StrBetween[1,2]{,\fir,}{,}{,}$};
		\node[vertex] (v2) at (1,1) {$\StrBetween[2,3]{,\fir,}{,}{,}$};
		\node[vertex] (v3) at (0,0) {$\StrBetween[3,4]{,\fir,}{,}{,}$};
		\node[vertex] (v4) at (1,0) {$\StrBetween[4,5]{,\fir,}{,}{,}$};
		\node[vertex] (v5) at (1.25,-.8) {$\StrBetween[5,6]{,\fir,}{,}{,}$};
		\draw[\sec] (v1) edge (v2) ;
		\draw[\thi] (v1) -- (v3);
		\draw[\fou] (v1) -- (v4);
                \draw[\fif, bend angle=10, bend right] (v1) edge (v5) ;
		\draw[\six] (v2) -- (v3);
		\draw[\sev] (v2) -- (v4);
                \draw[\eig, bend angle=20, bend left] (v2) edge (v5) ;
		\draw[\nin] (v3) -- (v4);
		\draw[#1, bend angle=5, bend right] (v3) edge (v5) ;
		\draw[#2] (v4) -- (v5) ;
\end{tikzpicture} }}

\let\FiveGraph=\FiveGraphR

\newcommand{\EightGraphContContCont}[2]{
\raisebox{-0.3\height}{\begin{tikzpicture}
\tikzstyle{vertex}=[circle,minimum size=10pt,inner sep=0pt]
	 \tikzstyle{g} = [draw,thick,dotted,-,green]
 	 \tikzstyle{r} = [draw,thick,-,red]
 	  	 \tikzstyle{u} = [draw,thick,-,white]
		\node[vertex] (v1) at (0,3) {$\StrBetween[1,2]{,\fir,}{,}{,}$};
		\node[vertex]  (v2) at (2,4) {$\StrBetween[2,3]{,\fir,}{,}{,}$};
		\node[vertex] (v3) at (2,2) {$\StrBetween[3,4]{,\fir,}{,}{,}$};
		\node[vertex] (v4) at (3,6) {$\StrBetween[4,5]{,\fir,}{,}{,}$};
		\node[vertex] (v5) at (3,0) {$\StrBetween[5,6]{,\fir,}{,}{,}$};
		\node[vertex] (v6) at (4,4) {$\StrBetween[6,7]{,\fir,}{,}{,}$};
		\node[vertex] (v7) at (4,2) {$\StrBetween[7,8]{,\fir,}{,}{,}$};
		\node[vertex] (v8) at (6,3) {$\StrBetween[8,9]{,\fir,}{,}{,}$};
		\draw[\sec]  (v1) -- (v2) ;
		\draw[\thi] (v1) -- (v3) ;
		\draw[\fou] (v1) -- (v4) ;
		\draw[\fif] (v1) -- (v5) ;
		\draw[\six] (v1) -- (v6) ;
		\draw[\sev] (v1) -- (v7) ;
		\draw[\eig] (v1) -- (v8) ;
		\draw[\nin] (v2) -- (v3) ;
		\draw[\ten] (v2) -- (v4) ;
		\draw[\ele] (v2) -- (v5) ;
		\draw[\twel] (v2) -- (v6) ;
		\draw[\thirt] (v2) -- (v7) ;
		\draw[\fourt] (v2) -- (v8) ;
		\draw[\fift] (v3) -- (v4) ;
		\draw[\sixt] (v3) -- (v5) ;
		\draw[\sevent] (v3) -- (v6) ;
		\draw[\eight] (v3) -- (v7) ;
		\draw[\ninet] (v3) -- (v8) ;
		\draw[\twe] (v4) -- (v5) ;
		\draw[\twefir] (v4) -- (v6) ;
		\draw[\twesec] (v4) -- (v7) ;
		\draw[\twethi] (v4) -- (v8) ;
		\draw[\twefou] (v5) -- (v6) ;
		\draw[\twefif] (v5) -- (v7) ;
		\draw[\twesix] (v5) -- (v8) ;
		\draw[\twesev] (v6) -- (v7) ;
		\draw[#1] (v6) -- (v8) ;
		\draw[#2] (v7) -- (v8) ;
\end{tikzpicture}}}
\newcommand{\EightGraphContCont}[9]{
	\def\ninet{#1}
	\def\twe{#2}
	\def\twefir{#3}
	\def\twesec{#4}
	\def\twethi{#5}
	\def\twefou{#6}
	\def\twefif{#7}
	\def\twesix{#8}
	\def\twesev{#9}
	\EightGraphContContCont
}
\newcommand{\EightGraphCont}[9]{
	\def\ten{#1}
	\def\ele{#2}
	\def\twel{#3}
	\def\thirt{#4}
	\def\fourt{#5}
	\def\fift{#6}
	\def\sixt{#7}
	\def\sevent{#8}
	\def\eight{#9}
	\EightGraphContCont
}
\newcommand{\EightGraph}[9]{
	\def\fir{#1}
	\def\sec{#2}
	\def\thi{#3}
	\def\fou{#4}
	\def\fif{#5}
	\def\six{#6}
	\def\sev{#7}
	\def\eig{#8}
	\def\nin{#9}
	\EightGraphCont
}

\newcommand{\NineGraphContContContCont}[1]{
\raisebox{-0.3\height}{\begin{tikzpicture}
		\node[vertex] (v1) at (0,3) {$\StrBetween[1,2]{,\fir,}{,}{,}$};
		\node[vertex] (v2) at (-0.6,2) {$\StrBetween[2,3]{,\fir,}{,}{,}$};
		\node[vertex] (v3) at (0.6,2) {$\StrBetween[3,4]{,\fir,}{,}{,}$};
		\node[vertex] (v4) at (-1.2,1) {$\StrBetween[4,5]{,\fir,}{,}{,}$};
		\node[vertex] (v5) at (-1.8,0) {$\StrBetween[5,6]{,\fir,}{,}{,}$};
		\node[vertex] (v6) at (-0.6,0) {$\StrBetween[6,7]{,\fir,}{,}{,}$};
		\node[vertex] (v7) at (1.2,1) {$\StrBetween[7,8]{,\fir,}{,}{,}$};
		\node[vertex] (v8) at (0.6,0) {$\StrBetween[8,9]{,\fir,}{,}{,}$};
		\node[vertex] (v9) at (1.8,0) {$\StrBetween[9,10]{,\fir,}{,}{,}$};
		\draw[\sec] (v1) -- (v2) ;
		\draw[\thi] (v1) -- (v3) ;
		\draw[\fou,bend angle=30,bend right] (v1) edge (v4) ;
		\draw[\fif,bend angle=40,bend right] (v1) edge (v5) ;
		\draw[\six] (v1) -- (v6) ;
		\draw[\sev,bend angle=30, bend left] (v1) edge (v7) ;
		\draw[\eig] (v1) -- (v8) ;
		\draw[\nin,bend angle=40, bend left] (v1) edge (v9) ;
		\draw[\ten] (v2) -- (v3) ;
		\draw[\ele] (v2) -- (v4) ;
		\draw[\twel,bend angle=30,bend right] (v2) edge (v5) ;
		\draw[\thirt] (v2) -- (v6) ;
		\draw[\fourt]   (v2) -- (v7) ;
		\draw[\fift]    (v2) -- (v8) ;
		\draw[\sixt]   (v2) -- (v9) ;
		\draw[\sevent] (v3) -- (v4) ;
		\draw[\eight]   (v3) -- (v5) ;
		\draw[\ninet]   (v3) -- (v6) ;
		\draw[\twe]    (v3) -- (v7) ;
		\draw[\twefir] (v3) -- (v8) ;
		\draw[\twesec,bend angle=30, bend left] (v3) edge (v9) ;
		\draw[\twethi] (v4) -- (v5) ;
		\draw[\twefou] (v4) -- (v6) ;
		\draw[\twefif] (v4) -- (v7) ;
		\draw[\twesix] (v4) -- (v8) ;
		\draw[\twesev] (v4) -- (v9) ;
	    \draw[\tweeig] (v5) -- (v6) ;
	    \draw[\twenin] (v5) -- (v7) ;
	    \draw[\thir, bend angle=30, bend right] (v5) edge (v8) ;
	    \draw[\thirfir, bend angle=40, bend right] (v5) edge (v9) ;
	    \draw[\thirsec] (v6) -- (v7) ;
	    \draw[\thirthi] (v6) -- (v8) ;
	    \draw[\thirfou, bend angle=30, bend right] (v6) edge (v9) ;
	    \draw[\thirfif] (v7) -- (v8) ;
	    \draw[\thirsix] (v7) -- (v9) ;
	    \draw[#1] (v8) -- (v9) ;
\end{tikzpicture}}}
\newcommand{\NineGraphContContCont}[9]{
    \def\tweeig{#1}
    \def\twenin{#2}
    \def\thir{#3}
    \def\thirfir{#4}
    \def\thirsec{#5}
    \def\thirthi{#6}
    \def\thirfou{#7}
    \def\thirfif{#8}
    \def\thirsix{#9}
    \NineGraphContContContCont
}
\newcommand{\NineGraphContCont}[9]{
	\def\ninet{#1}
	\def\twe{#2}
	\def\twefir{#3}
	\def\twesec{#4}
	\def\twethi{#5}
	\def\twefou{#6}
	\def\twefif{#7}
	\def\twesix{#8}
	\def\twesev{#9}
	\NineGraphContContCont
}
\newcommand{\NineGraphCont}[9]{
	\def\ten{#1}
	\def\ele{#2}
	\def\twel{#3}
	\def\thirt{#4}
	\def\fourt{#5}
	\def\fift{#6}
	\def\sixt{#7}
	\def\sevent{#8}
	\def\eight{#9}
	\NineGraphContCont
}
\newcommand{\NineGraph}[9]{
	\def\fir{#1}
	\def\sec{#2}
	\def\thi{#3}
	\def\fou{#4}
	\def\fif{#5}
	\def\six{#6}
	\def\sev{#7}
	\def\eig{#8}
	\def\nin{#9}
	\NineGraphCont
}

\renewcommand{\phi}{\varphi}
\renewcommand{\emptyset}{\varnothing}
\renewcommand{\epsilon}{\varepsilon}

\newcommand{\NP}{\mathbf{NP}}
\newcommand{\co}{\mathbf{co}}
\newcommand{\coNP}{\co\NP}

\newcommand{\infers}{\rightarrow}
\newcommand{\Lin}{ L}

\newcommand{\V}{\ensuremath{\mathcal{V}}}
\newcommand{\W}{\ensuremath{\mathcal{W}}}
\newcommand{\R}{\ensuremath{\mathcal{R}}}
\renewcommand*{\S}{\ensuremath{\mathcal{S}}}
\newcommand{\m}{\ensuremath{\mathsf{m}}}
\newcommand{\s}{\ensuremath{\mathsf{s}}}

\newcommand{\ms}{\ensuremath{\mathsf{ms}}}
\newcommand{\msu}{\ensuremath{\mathsf{msu}}}
\newcommand{\ac}{\ensuremath{\mathsf{ac}}}
\newcommand{\un}{\ensuremath{\mathsf{u}}}
\newcommand{\acu}{\ensuremath{\mathsf{acu}}}
\newcommand{\red}{\ensuremath{\overset{*}{\rightsquigarrow}}}
\newcommand{\redms}{\ensuremath{\red_\ms}}
\newcommand{\redmsu}{\ensuremath{\red_\msu}}
\newcommand{\reds}{\ensuremath{\red_\s}}

\newcommand{\redsu}{\ensuremath{\red_{\s\un}}}

\newcommand{\basis}{\ensuremath{\mathsf{basis}}}

\newcommand{\pp}[2]{x_{#1#2}}
\newcommand{\qq}[2]{y_{#1#2}}
\newcommand{\rr}[2]{z_{#1#2}}

\Crefname{prop}{Proposition}{Propositions}
\Crefname{exa}{Example}{Examples}
\Crefname{cor}{Corollary}{Corollaries}
\Crefname{rem}{Remark}{Remarks}
\Crefname{lem}{Lemma}{Lemmas}
\Crefname{defi}{Definition}{Definitions}
\Crefname{thm}{Theorem}{Theorems}
\Crefname{fact}{Fact}{Facts}
\Crefname{obs}{Observation}{Observations}
\Crefname{fig}{Figure}{Figures}

\theoremstyle{defC}\newtheorem{exaC}[thm]{Example}

\title[Enumerating independent linear inferences]{Enumerating independent linear inferences\rsuper*} 
\titlecomment{{\lsuper*}This is an extended version of \cite{DasRice21}, published in the proceedings of FSCD'21.}

\author[A.~Das]{Anupam Das\lmcsorcid{0000-0002-0142-3676}}[a]
\address{University of Birmingham}
\email{a.das@bham.ac.uk}

\author[A.~Rice]{Alex Rice\lmcsorcid{0000-0002-2698-5122}}[b]
\address{University of Cambridge}
\email{alex.rice@cl.cam.ac.uk}

\thanks{The authors would like to thank Matteo Acclavio, Ross Horne, Lutz Stra{\ss}burger, and Alessio Guglielmi for several comments and discussions on this and earlier versions of this work.}

\begin{document}

\begin{abstract}
  A \emph{linear inference} is a valid inequality of Boolean algebra in which each variable occurs at most once on each side.

  In this work we leverage recently developed \emph{graphical} representations of linear formulae to build an implementation that is capable of more efficiently searching for switch-medial-independent inferences. We use it to find four `minimal' 8-variable independent inferences and also prove that no smaller ones exist; in contrast, a previous approach based directly on formulae reached computational limits already at 7 variables. Two of these new inferences derive some previously found independent linear inferences. The other two (which are dual) exhibit structure seemingly beyond the scope of previous approaches we are aware of; in particular, their existence contradicts a conjecture of Das and Stra{\ss}burger.

  We were also able to identify 10 minimal 9-variable linear inferences independent of all the aforementioned inferences, comprising 5 dual pairs, and present applications of our implementation to recent `graph logics'.
\end{abstract}

\maketitle

\section{Introduction}
\label{sec:introduction}
A \emph{linear inference} is a valid implication $\phi \infers \psi$ of Boolean logic, where $\phi$ and $\psi$ are \emph{linear}, i.e.\ each variable occurs at most once in each of $\phi$ and $\psi$.
Such implications have played a crucial role in many areas of structural proof theory.
For instance the inference \emph{switch},
\[
\s \ : \
x \land (y \lor z)
\ \infers \
(x \land y) \lor z
\]
governs the logical behaviour of the \emph{multiplicative} connectives $\parr$ and $\otimes$ of linear logic \cite{Gir87:linear-logic},
and similarly the inference \emph{medial},
\[
\m \ : \
(w \land x) \lor (y \land z)
\ \infers \
(w \lor y) \land (x \lor z)
\]
together with the structural rules \emph{weakening} and \emph{contraction}, governs the logical behaviour of the \emph{additive} connectives $\oplus$ and $\&$ \cite{Str02:lin-log-loc-sys,Str03:phd-thesis}.
Both of these inferences are fundamental to \emph{deep inference} proof theory, in particular allowing weakening and contraction to be reduced to atomic form \cite{BruTiu01:class-log-loc-sys,Bru03:phd-thesis}, thereby admitting elegant `geometric' proof normalisation procedures based on \emph{atomic flows} \cite{GugGun08:norm-contr-di-at-fl,Gun09:phd-thesis}.
One particular feature of these normalisation procedures is that they are robust under the addition of further linear inferences to the system, thanks to the atomisation of structural steps.

On the other hand the set of \emph{all} linear inferences $\Lin$ plays an essential role in certain {models} of linear logic and related substructural logics.
In particular, the multiplicative fragment of Blass' \emph{game semantics} model of linear logic validates \emph{just} the linear inferences (there called `binary tautologies') \cite{Bla92:game-semantics-ll}, and this coincides too with the multiplicative fragment of Japaridze's \emph{computability logic}, cf., e.g., \cite{Jap05:intro-to-cirquent-calc}.
From a complexity theoretic point of view, the set $\Lin$ is sufficiently rich to encode all of Boolean logic: it is $\coNP$-complete \cite{Str12:ext-wo-cut,DasStr16:no-compl-lin-sys}.

It was recently shown by one of the authors, together with Stra{\ss}burger, that, despite its significance, $\Lin$ admits no polynomial-time axiomatisation by linear inferences unless $\coNP = \NP$ \cite{DasStr15:no-comp-lin-sys,DasStr16:no-compl-lin-sys}, resolving a long-standing open problem of Blass and Japaridze for their respective logics (see, e.g., \cite{Jap17:elem-base-cirq-calc}).
Note that the condition of polynomial-time computability naturally arises from proof theory \cite{CooRec74:length-of-proofs,CooRec79:rel-eff-pps}, and is also required for the result to be meaningful: it prevents us just taking the entire set $\Lin$ as an axiomatisation.

From a Boolean algebra point of view, this result means that the class of linear Boolean inequalities has no feasible basis (unless $\coNP = \NP$).
From a proof theoretic point of view this means that any propositional proof system (in the Cook-Reckhow sense \cite{CooRec74:length-of-proofs,CooRec79:rel-eff-pps}, see also \cite{Kra19:cook-reckhow}) must necessarily admit some `structural' behaviour, even when restricted to proving only linear inferences (unless $\coNP= \NP$).

An immediate consequence of this result is that $\s$ and $\m$ above do not suffice to generate all linear inferences (unless $\coNP= \NP$), even modulo all valid linear equations (which are just associativity, commutativity, and unit laws, cf.~\cite{DasStr15:no-comp-lin-sys,DasStr16:no-compl-lin-sys}).
In fact, this was known before the aforementioned result, due to the identification of an explicit 36 variable inference in \cite{Str12:ext-wo-cut}.\footnote{Stra{\ss}burger refers to the inference as a `balanced tautology', but like the `binary tautologies' of Blass and Japaridze, these are equivalent to linear inferences. In particular we recast Stra{\ss}burger's example as a bona fide linear inference in \cref{sec:prev-lin-infs}.}
Already in that work the question was posed whether such an inference was minimal, and since then the identification of a minimal $\{\s,\m \}$-independent linear inference has been a recurring theme in the literature of this area.

It has been verified in \cite{Das13:lin-inf-rew} that a minimal $\{\s,\m \}$-independent linear inference must be `non-trivial', as long as we admit all valid linear equations (again, these are just associativity, commutativity and unit laws).
Intuitively, `non-triviality' rules out pathological inferences such as $x \land y \infers x \lor y$ or $x \land (y \lor z) \infers x \lor (y \land z)$. For these inferences the variable, say, $y$ is, in a sense, redundant; it turns out that they may be derived in $\{\s,\m\}$, modulo linear equations, from a smaller non-trivial `core'.
We recall these arguments in \cref{sec:preliminaries}.

Furthermore \cite{Das13:lin-inf-rew} identified a 10 variable linear inference that is not derivable by switch and medial (even under linear equations), which Stra{\ss}burger conjectured was minimal \cite{Str12:private-conjecture}.
Around the same time \v{S}ipraga attempted a computational approach, searching for independent linear inferences by brute force \cite{Sip12:aut-search-lin-inf}.
However, computational limits were reached already at 7 variables.
In particular, every linear inference of up to 6 variables is already derivable by switch and medial, modulo linear equations; due to the aforementioned 10 variable inference, any minimal independent linear inference must have size 7,8,9, or 10.

Since 2013 there have been significant advances in the area, in particular through the proliferation of \emph{graph-theoretic} tools.
Indeed, the interplay between formulae and graphs was heavily exploited for the aforementioned result of \cite{DasStr15:no-comp-lin-sys,DasStr16:no-compl-lin-sys}.
Since then, multiple works have emerged in the world of linear proof theory that treat these graphs as `first class citizens', comprising a now active area of research \cite{NguSei18:coh-int-graphs,AccHorStr20:mll-graphs-short,AccHorStr20:mll-graphs-full,CalDasWar20:bgl}.

\paragraph*{Contribution}
In this work we revisit the question of minimal $\{\s,\m\}$-independent linear inferences by exploiting the aforementioned recent graph theoretic techniques.
Such an approach vastly reduces the computational resources necessary and, in particular, we are able to provide a conclusive result:
the smallest $\{\s,\m \}$-independent linear inference has size 8. In fact there are four minimal (w.r.t.\ derivability) such ones, unique up to associativity, commutativity, renaming of variables, given by,
\begin{equation}
  \label{eq:php32-derived-inf}
  \begin{alignedat}{2}
    & &&(z \lor (w \land w')) \land ((x \land x' ) \lor ((y \lor y') \land z'))\\
    & \to &\quad&(z \land (x \lor y)) \lor ((w \lor y') \land ((w' \land x')  \lor z') )
  \end{alignedat}
\end{equation}
\\
\begin{equation}
    \label{eq:unknown-inference-8}
    \begin{alignedat}{2}
    & &&(z \lor (w \land w')) \land ((x \land x') \lor ((y \lor y') \land z'))\\
    & \to &\quad&((z \lor (w \land x)) \land (x' \lor y)) \lor (z' \land (w' \lor y'))
    \end{alignedat}
\end{equation}
\\
\begin{equation}
  \label{eq:counterexample-inference}
  \begin{alignedat}{2}
    & &&((w \land w') \lor (x \land x')) \land ((y \land y') \lor (z \land z'))\\
    & \to&\quad& (w \land y) \lor ((x \lor (w'\land z')) \land ((x'\land y') \lor z) )
  \end{alignedat}
\end{equation}
and the De Morgan dual of \cref{eq:counterexample-inference}.
Note that \cref{eq:unknown-inference-8} did not appear in the preliminary version \cite{DasRice21} as it was erroneously considered to be isomorphic to the first.
We dedicate some discussion to each of these separately in \cref{sec:two-found-inferences}, and include a manual verification of their soundness and $\{\s,\m \}$-independence in \cref{sect:app:further-proofs-examples}, as a sanity check.

Our main contribution, as announced in the preliminary version of this paper \cite{DasRice21}, is an implementation that checks inference for \(\{\s,\m\}\)-derivability, which was able to confirm that all 7 variable linear inference are derivable from switch and medial. In fact we found \eqref{eq:php32-derived-inf} independently of the implementation presented in this paper.\footnote{These two developments were respectively communicated via blog posts \cite{Ric20:lin-inf-size-7} and \cite{Das20:lin-inf-size-8}.}
Ultimately, we improved the implementation to run on inferences of size 8 too,
and our inference \eqref{eq:php32-derived-inf} was duly found, as well as \eqref{eq:unknown-inference-8}, \eqref{eq:counterexample-inference}, and the dual of \eqref{eq:counterexample-inference}.
One highlight of \eqref{eq:counterexample-inference} is that it exhibits a peculiar structural property that \emph{refutes} Conjecture 7.9 from \cite{DasStr16:no-compl-lin-sys}, as we explain in \cref{sec:counterexample-inference}.

In this extended version, we also show how our search algorithm can be generalised to enumerate minimal linear inferences independent of an arbitrary set $S$ of linear inferences.
 This can be employed to devise a recursive algorithm $\basis(n)$ that computes all $n$-variable minimal linear inferences independent of all smaller ones.
To this end, in \cref{sec:further-theoretical} we define the notion of an $n$-minimal set, which is intuitively a minimal set of inferences that derives all $n$-variable inferences modulo all $<n$-variable inferences.
We show these $n$-minimal sets are unique up to linear equations and renaming of variables, and finally prove correctness of our algorithm $\basis(n)$ by extending some theoretical results on constant-freeness for switch and medial to the general setting of arbitrary $n$-minimal sets.
For instance the 9-variable minimal linear inferences given by $\basis(9)$ are those not derivable from switch, medial, and the four aforementioned 8-variable linear inferences, instead of the weaker condition of not being derivable from only switch and medial.
Indeed, we were able to successfully run $\basis(9)$ to compute the $9$-minimal linear inferences.
We classify these in \cref{sec:9var-infs}, identifying 10 such inferences in total, comprising 5 dual pairs.

Our implementation \cite{Ric21:implementation} is split into a \emph{library} and an \emph{executable}, where the executable implements our search algorithm described in \cref{sec:search}, and the library contains foundations for working with linear inferences using the graph theoretic techniques presented in \cref{sec:webs}. These are written in Rust and designed to be relatively fast while maintaining readability.

In addition to having the machinery needed for our original search, the library has also been extended for to work with \emph{arbitrary graphs}, e.g.\ as presented in \cite{CalDasWar20:bgl}, in particular being able to check other notions of inference for graphs and having tools for determining whether an inference is an instance of an arbitrary, user-supplied graph rewrite.
Our intention is that this could form a reusable base for future investigations in the area, both for linear formulae and for the recent linear graph theoretic settings of \cite{NguSei18:coh-int-graphs,AccHorStr20:mll-graphs-short,AccHorStr20:mll-graphs-full,CalDasWar20:bgl}.

As well as the additions already mentioned, this extended version contains several further proofs, examples and general discussion compared with the preliminary version \cite{DasRice21}.
\section{Preliminaries}
\label{sec:preliminaries}
In this section we present some preliminaries on linear inferences, triviality and minimality.
The content of this section is based on Section 2 of the preliminary version, \cite{DasRice21}, with some additional proofs and narrative.

Throughout this paper we shall work with a countably infinite set of \textbf{variables}, written $x,y, z$ etc.
A \textbf{linear formula} on a (finite) set of variables \(\V\) is defined recursively as follows:
\begin{itemize}

\item \(\top\) and \(\bot\) are linear formulae on $\emptyset$, the empty set of variables (called \textbf{units} or \textbf{constants}).

\item \(x\) and \(\neg x\) are linear formulae on \(\{x\}\), for each variable $x$.
\item If \(\phi\) is a linear formula on \(\V_1\) and \(\psi\) is a linear formula on \(\V_2\), with \(\V_1 \cap \V_2  = \emptyset \), then \(\phi \lor \psi\) and \(\phi \land \psi\) are linear formulae on \(\V_1 \cup \V_2\).
\end{itemize}

We adopt the expected Boolean semantics of (linear) formulae. Namely we extend any map $\alpha$ from variables to $\{0,1\}$ (an `\textbf{assignment}') to one on all (linear) formulae as follows:
\[
\begin{array}{r@{\ := \ }l}
     \alpha(\top) & 1 \\
     \alpha (\bot) & 0 \\
     \alpha (\neg x ) & 1 - \alpha (x) \\
\end{array}
\qquad
\begin{array}{r@{\ := \ }l}
     \alpha(\phi \land \psi) & \min (\alpha (\phi) , \alpha(\psi)) \\
     \alpha(\phi \lor \psi) & \max (\alpha (\phi) , \alpha(\psi))
\end{array}
\]
The formula $\phi$ \textbf{satisfies} an assignment $\alpha$ if $\alpha (\phi)=1$.
Formulae $\phi $ and $\psi$ are \textbf{logically equivalent} if they satisfy the same assignments.

\begin{rem}
  [On negation]
  Note that our restriction of negation to only propositional variables is made without loss of generality. Namely, we may recover $\neg \phi$ thanks to De Morgan's laws by setting:
  \[
  \neg \bot := \top
  \qquad
  \neg \top := \bot
  \qquad
  \neg \neg \phi := \phi
  \qquad
  \neg (\phi \lor \psi) := \neg \phi \land \neg \psi
  \qquad
  \neg (\phi \land \psi) := \neg \phi \lor \neg \psi
  \]
  Note in particular that, if $\phi$ is a linear formula on $\V$, then so is $\neg \phi$.
\end{rem}

A linear formula that does not contain \(\top\) or \(\bot\) is \textbf{constant-free}.
A linear formula with no negated variables (i.e.\ formulae of form $\neg x$) is \textbf{negation-free}. Later in the paper, we will be able to restrict our search to inferences between constant-free negation-free formulae.

In what follows, we shall omit explicit consideration of variable sets, assuming that they are disjoint whenever required by the notation being used.

A binary relation \(\sim\) on linear formulae is \textbf{closed under contexts} if we have for all \(\phi,\psi\), and \(\chi\):
\begin{alignat*}{2}
  &\phi \sim \psi \implies \phi \land \chi \sim \psi \land \chi &\quad\quad&\phi \sim \psi \implies \phi \lor \chi \sim \psi \lor \chi\\
  &\phi \sim \psi \implies \chi \land \phi \sim \chi \land \psi &&\phi \sim \psi \implies \chi \lor \phi \sim \chi \lor \psi
\end{alignat*}
An equivalence relation (on linear formulae) that is closed under contexts is called a (linear) \textbf{congruence}.

\begin{defi}
[Linear equations]
Let \(\sim_\ac\) be the smallest congruence satisfying,
\begin{alignat*}{2}
  &\phi \lor \psi \sim_\ac \psi \lor \phi &\quad\quad&\phi \land (\psi \land \chi) \sim_\ac (\phi \land \psi) \land \chi\\
  &\phi \land \psi \sim_\ac \psi \land \phi &&\phi \lor (\psi \lor \chi) \sim_\ac (\phi \lor \psi) \lor \chi
\end{alignat*}

$\sim_\un$ is the smallest congruence satisfying:
\begin{equation}
  \label{eq:un-eq-rel}
  \begin{alignedat}{4}
    &\phi \land \top \sim_\un \phi   &\quad& \phi \lor \bot\sim_\un \phi &\quad&
    \top \land \phi \sim_\un  \phi  &\quad&  \bot \lor \phi \sim_\un \phi  \\
    &\phi \land \bot \sim_\un \bot  && \phi \lor \top \sim_\un \top  &&
    \bot \land \phi \sim_\un \bot  && \top \lor \phi \sim_\un \top
  \end{alignedat}
\end{equation}

\(\sim_\acu\) is the smallest congruence containing both $\sim_\ac$ and $\sim_\un$.
\end{defi}
\begin{rem}
[Variable sets under congruence]
Suppose $\phi$ and $\phi'$ are linear formulae on $\V $ and $\V'$ respectively.
If $\phi \sim_\ac \phi'$ then $\V = \V'$.
\end{rem}
Note also that $\sim_\un$ generates a unique normal form of linear formulae by maximally eliminating constants:

\begin{prop}
[Folklore]
  \label[prop]{prop:unit-free}
  Every formula is $\sim_\un$-equivalent to a unique constant-free formula, or is equivalent to $\bot$ or $\top$.
\end{prop}
\begin{proof}
Construe the equations of \cref{eq:un-eq-rel} as a rewrite system by orienting them left-to-right.
This system reduces the size of formulae and so must terminate.
By inspection, the normal form must contain no constant under the scope of a connective $\lor $ or $\land$.

For uniqueness, it suffices to consider the critical pairs. In fact, the analysis is particularly simple since any critical pair is a $\lor$ or $\land $ of constants that always reduces to the same constant:
\[
\begin{array}{rcccl}
    \top & \underset{\top \land }{\longleftarrow} & \top \land \top & \underset{ \land \top }{\longrightarrow} & \top \\
    \bot & \underset{\top \land }{\longleftarrow} &  \top \land \bot & \underset{ \land \bot }{\longrightarrow} &  \bot \\
    \bot & \underset{\bot \land }{\longleftarrow} &\bot \land \top & \underset{ \land \top }{\longrightarrow} & \bot\\
    \bot &\underset{\bot \land }{\longleftarrow} &\bot \land \bot& \underset{ \land \bot }{\longrightarrow} &\bot
\end{array}
\qquad\qquad
\begin{array}{rcccl}
    \top & \underset{\top \lor }{\longleftarrow} & \top \lor \top &  \underset{ \lor \top }{\longrightarrow} & \top\\
    \top & \underset{\top \lor }{\longleftarrow} & \top \lor \bot & \underset{ \lor \bot }{\longrightarrow} & \top \\
    \top & \underset{\bot \lor }{\longleftarrow} & \bot \lor \top & \underset{ \lor \top }{\longrightarrow} & \top\\
    \bot & \underset{\bot \lor }{\longleftarrow} & \bot \lor \bot & \underset{ \lor \bot }{\longrightarrow} & \bot
\end{array}
\]
For the above pairs, we have indicated a reduction, say $\top \land \phi \to \phi$ or $\phi \lor \top \to \top$, by $\top \land $ or $\lor \top$ respectively.
\end{proof}

\begin{rem}
[On logical equivalence]
\label[rem]{acu-and-logical-equivalence}
Clearly, if $\phi \sim_\acu \psi$ then $\phi$ and $\psi$ are logically equivalent.
In fact, for linear formulae, we also have a converse: two linear formulae $\phi$ and $\psi$ on the same variables are logically equivalent if and only if $\phi \sim_\acu \psi$ \cite{DasStr15:no-comp-lin-sys,DasStr16:no-compl-lin-sys}. Furthermore, given two linear formulae \(\phi\) and \(\psi\) that are constant-free, \(\bot\), or \(\top\), then \(\phi\) and \(\psi\) are logically equivalent if and only if they are $\sim_\ac$-equivalent.
These property follows from \cref{prop:unit-free} above, the results of \cref{sec:trivial}, and the graphical representation of linear formulae and their semantics in \cref{sec:webs}.
\end{rem}

\subsection{Linear inferences}
A \textbf{linear inference}
is just a valid implication $\phi \to \psi$ (i.e., $\alpha (\phi) = 1 \implies \alpha (\psi) = 1$) where $\phi$ and $\psi$ are linear formulae.
The left-hand side (LHS) and right-hand side (RHS) of a linear inference, generally speaking, need not be linear formulae on the same variables.
Nonetheless we shall occasionally refer to linear inferences ``on $\V$'' or ``on \(n\) variables'', assuming that the LHS and RHS are both linear formulae on some fixed \(\V\) with \(|\V| = n\).

There are two linear inferences we shall particularly focus on, due to their prevalence in structural proof theory. \textbf{Switch} is the following inference on 3 variables,
\begin{equation}\label{eq:switch}
\s : x \land (y \lor z) \to (x \land y) \lor z
\end{equation}
and \textbf{medial} is the following inference on 4 variables:
\begin{equation}\label{eq:medial}
\m : (w \land x) \lor (y \land z) \to (w \lor y) \land (x \lor z)
\end{equation}

We may compose switch and medial (and more generally an arbitrary set of linear inferences) to form new linear inferences by construing them as \emph{term rewriting rules}.
More generally, we will consider rewriting derivations modulo the equivalence relations $\sim_\ac$ and $\sim_\acu$ we introduced earlier.
In the latter case, as previously mentioned, the underlying set of variables may change during a derivation, though \cref{prop:unit-free} will later allow us to work with some fixed set of variables throughout $\{\s, \m\}$ derivations.

\begin{defi}[Rewriting]
Let $S$ be a set of linear inferences.
We write $\to_S$ for the term rewrite system generated by $S$. I.e., $\to_S$ is the smallest relation containing each inference of $S$ and closed under substitution and contexts.

We write $\rightsquigarrow_S$ for the closure of $\to_S$ modulo $\ac$, i.e.\ $\phi \rightsquigarrow_S \psi$ if there are $\phi',\psi' $ such that $\phi \sim_\ac \phi' \to_S \psi' \sim_\ac \psi$. In this case we may say that the inference $\phi \to \psi$ is an \textbf{instance (without units)} of some inference in $S$.
We write $\rightsquigarrow_{S\un}$ for the closure of $\to_S$ modulo $\acu$, i.e.\ $\phi \rightsquigarrow_{S\un} \psi$ if there are $\phi',\psi' $ such that $\phi \sim_\acu \phi' \to_S \psi' \sim_\acu \psi$. Similarly in this case we may simply say the inference is an \textbf{instance (with units)} of some inference in $S$.

We write $\red_S$, or $\red_{S\un}$, for the smallest reflexive transitive relation containing $\rightsquigarrow_S$ and $\sim_\ac$, or $\rightsquigarrow_{S\un}$ and $\sim_\acu$, respectively, and say that a linear inference $\phi \to \psi $ is \textbf{$S$-derivable (without units)} (or \textbf{$S$-derivable with units}) if $\phi \red_S \psi$ (or $\phi \red_{X\un} \psi$, respectively).
\end{defi}

We shall typically write such sets $S$ by just listing their linear inferences, e.g.\ we write $\redms$ instead of $\red_{\{\m,\s\}}$ and $\redmsu$ instead of $\red_{\{\m,\s \}\un}$, in order to lighten notation.
Clearly, $\s$ and $\m$ are {valid}, so any derivation $\phi \redmsu \psi$ comprises a linear inference.

\begin{exa}
[Weakening and duality]
\label[exa]{example-weakening-duality}
The following derives \(\bot \redsu \chi\) for any \(\chi\):
\[\bot \sim_\acu \bot \land (\top \lor \chi) \to_\s (\bot \land \top) \lor \chi \sim_\acu \chi \]

Using this, we may $\{\s\}$-derive \textbf{weakening}, $\phi \to \phi \lor \chi$, with units as follows:
\[
\phi \sim_\acu \phi \lor \bot \redsu \phi \lor \phi \lor \chi
\]
Notice that $\redsu$ is closed under De Morgan duality:
If $\phi \redsu \chi$ and $\bar \phi$ and $\bar \chi$ are obtained from $\phi$ and $\chi$, respectively, by flipping each $\lor$ to a $\land$ and vice versa, then $\bar \chi\redsu \bar \phi$.
This follows by direct inspection of $\s$ and each clause of $\sim_\acu$; indeed the same property holds for $\reds$,$\redms$, and $\redmsu$ by the same reasoning.
As a result, we also have that \textbf{coweakening}, $\phi \land \chi \to \phi$, is $\{\s\}$-derivable with units.
\end{exa}

\begin{exa}
[`Mix']
\label[exa]{ex:mix}
Units can help us derive even constant-free linear inferences.
For instance, \textbf{mix}: \(\phi \land \psi \to \phi \lor \psi\) can be derived as a composite of coweakening and weakening:
\[\phi \land \chi \redsu \phi \redsu \phi \lor \chi \]
  Note that mix is not derivable without using instances of $\sim_\un$.
\end{exa}

The following is one of the main results of this paper:
\begin{thm}
  \label{thm:main}
  Suppose \(\phi\) is a linear formula over \(\V_1\) and \(\psi\) is a linear formula over \(\V_2\) and \(\mathsf r : \phi \to \psi\) is a linear inference. Then if \(|\V_1 \cap \V_2| \leq 7\) we have that \(\phi \redmsu \psi\).

  Furthermore, there is a valid linear inference \(\phi \to \psi\) on \(8\) variables with \(\phi \not \redmsu \psi\), so \(7\) is maximal with the property above.
\end{thm}

This theorem was the original objective of the preliminary version of this paper~\cite{DasRice21}. Indeed, before this it was not known if the smallest linear inference independent of switch and medial used 7 or 8 variables. The rest of \cref{sec:preliminaries} and \cref{sec:webs,sec:algorithm} work towards a proof of this theorem.

\subsection{Trivial inferences}
\label{sec:trivial}
In order to state \cref{thm:main} above in its most general form, we have allowed linear formulae to include constants and negation, and linear inferences to be between formulae with different variable sets.
However
it turns out that we may proceed to prove \cref{thm:main}, without loss of generality, by working with derivations of only constant-free, negation-free formulae on some fixed set of variables, as was already shown in \cite{Das13:lin-inf-rew}. This is done by defining the notion of a \emph{trivial} inference, whose $\{\s,\m \}$-derivability, with units, may be reduced to that of a smaller non-trivial inference.

\begin{defi}
  An inference \(\phi \to \psi\) is \textbf{trivial at a variable \(x\)} if \(\phi[\top/x] \to \psi[\bot/x]\) is again a valid inference.
  An inference is \textbf{trivial} if it is trivial at one of its variables.
\end{defi}

\begin{exa}
  The mix inference from \cref{ex:mix}, \(x \land y \to x \lor y\), is trivial at \(x\) and trivial at \(y\).
  Note, however, that it is not trivial at $x$ and $y$ `at the same time', in the sense that the simultaneous substition of $\bot $ for $x$ and $y$ in the LHS and $\top$ for $x$ and $y$ in the RHS does not result in a valid implication.
	In contrast, the linear inference $w \land (x \lor y) \to w \lor (x \land y)$ from \cite{Das13:lin-inf-rew} is, indeed, trivial at $x$ and $y$ {`at the same time'}.

  Neither switch nor medial are trivial.
\end{exa}

\begin{rem}
[Global vs local triviality]
\label[rem]{trivial-composition}
Note that triviality is closed under composition by linear inferences: if $\phi \to \psi$ is trivial at $x$ and $\psi \to \chi$ is valid, then $\phi \to \chi$ is trivial at $x$.
Similarly for $\chi \to \psi$ if $\chi \to \phi$ is valid.
One pertinent feature is that the converse does not hold: there are `globally' trivial derivations that are nowhere `locally' trivial.
For instance consider the following derivation (from \cite[Remark~5.6]{DasStr16:no-compl-lin-sys}):
\[
w \land x \land (y \lor z)
\rightsquigarrow_\s
w \land ((x \land y) \lor z)
\rightsquigarrow_\s
(w\land z) \lor (x \land y)
\rightsquigarrow_\m
(w\lor x) \land (y \lor z)
\]
The derived inference is just an instance of mix, from \cref{ex:mix}, on the redex $w \land x$, which is trivial.
However, no local step is trivial.
\end{rem}

To prove (the first half of) \cref{thm:main}, in \cref{sec:algorithm} we will actually prove the following apparent weakening of that statement:

\begin{thm}
  \label{thm:main-reduced}
  Let \(n < 8\). Let \(\phi\) and \(\psi\) be constant-free negation-free linear formulae on \(n\) variables and suppose \(\phi \to \psi\) is a non-trivial linear inference. Then \(\phi \redms \psi\).
\end{thm}
In fact this statement is no weaker at all, and we will now see how the consideration of triviality allows us to only deal with such special cases without loss of generality. We first show that is not necessary to consider linear inferences where the premise and conclusion have different variables.

\begin{prop}
\label[prop]{prop:diff-vars}
    Let $\phi$ and $\psi$ be linear formulae on $\V_1$ and $\V_2$, respectively, and let $\mathsf r : \phi \to \psi$ be a linear inference. Then there is $\mathsf r' : \phi' \to \psi'$ on $\V_1 \cap \V_2$ such that $\mathsf r$ is $\{\s,\m,\mathsf r'\}$-derivable with units.
\end{prop}
\begin{proof}
    First note that by setting $\phi = \top$ in \Cref{ex:mix} we have for any $\psi$:
    \[ \psi \sim_\acu \top \land \psi \redmsu \top \lor \psi \sim_\acu \top \]
    and by setting $\psi = \bot$ we have for all $\phi$ that:
    \[ \bot \sim_\acu \phi \land \bot \redmsu \phi \lor \bot \sim_\acu \phi \]

    Now we can let $\phi'$ be $\phi$ with every variable in $\V_1 \setminus \V_2$ (or its negation) replaced with $\top$, and $\psi'$ be $\psi$ with every variable in $\V_2 \setminus \V_1$ (or its negation) replaced with $\bot$. Then $\mathsf r' : \phi' \to \psi'$ is clearly a linear inference (since $\mathsf r$ is valid w.r.t.\ all assignments) and only uses the variables $\V_1 \cap \V_2$. Furthermore we have the derivation:
    \[ \phi \redmsu \phi' \rightsquigarrow_{\mathsf r'} \psi' \redmsu \psi\]
    where the first step applies the derivation of \(x \redmsu \top\) (or $\neg x \redmsu \top$) to each variable \(x \in \V_1 \setminus \V_2\) and the last step applies the derivation of \(\bot \redmsu y\) (or $\bot \redmsu \neg y$) to each variable \(y \in \V_2 \setminus \V_1\).
\end{proof}

The key use of triviality is that if an inference is trivial, then there is a smaller non-trivial inference from which it can be derived.
\begin{propC}
[{\cite[Theorem~34]{Das13:lin-inf-rew}}]
  \label[prop]{prop:trivial}
  Let $\phi$ and $\psi$ be linear formulae on $\V$,
  and let \(\mathsf r : \phi \to \psi\) be a trivial linear inference.
  Then there is a non-trivial linear inference \(\mathsf r' : \phi' \to \psi'\) on some $\V' \subsetneq \V$ such that $\mathsf r : \phi \to \psi$ is $\{\s,\m,\mathsf r' \}$-derivable with units.
\end{propC}

Note in particular that, in both statements above, if $\mathsf r'$ is $\{\s,\m \}$-derivable with units, then so is $\mathsf r$.
This is also the case for the next result.

\begin{prop}
  \label[prop]{prop:const-free-neg-free}
  Let \(\mathsf r: \phi \to \psi\) be a non-trivial linear inference among variables $\V \neq \emptyset$.
  Then there is a constant-free negation-free non-trivial linear inference \(\mathsf r' : \phi' \to \psi'\) on $\V$
  s.t.\ $\mathsf r: \phi \to \psi$ is $\{\s,\m,\mathsf r' \}$-derivable with units.
\end{prop}
\begin{proof}
  First, note that both $\phi$ and $\psi$ must be linear formulae on $\V$, since $\phi \to \psi$ is non-trivial.
  For the same reason, no variable can occur positively in $\phi$ and negatively in $\psi$ or vice-versa, since $\phi \to \psi$ is non-trivial, and so any negated variable may be safely replaced by its positive counterpart.
From here, we simply set $\phi'$ and $\psi' $ to be the constant-free formulae (uniquely) obtained from \cref{prop:unit-free} by $\sim_\un$.
Non-triviality of $\mathsf r'$ follows from that of $\mathsf r$ by logical equivalence.
\end{proof}

\begin{cor}\label[cor]{cor:main-red-to-main}
  The statement of \cref{thm:main-reduced} implies (the first half of) the statement of \cref{thm:main}.
\end{cor}
\begin{proof}
Let $\mathsf r$ be as in \cref{thm:main}.
Let $\mathsf r'$ be the inference where the LHS and RHS have the same variable set obtained from \cref{prop:diff-vars}, let $\mathsf r''$ be the non-trivial linear inference obtained by \cref{prop:trivial} above, and finally let $\mathsf r'''$ be the non-trivial constant-free negation-free linear inference obtained by \cref{prop:const-free-neg-free}.
By \cref{thm:main-reduced}, $\mathsf r'''$ is $\{\s,\m \}$-derivable and so, by \cref{prop:diff-vars}, \cref{prop:trivial}, and \cref{prop:const-free-neg-free}, $\mathsf r$ is also $\{\s,\m \}$-derivable with units.
\end{proof}

It is clear that if an inference is derivable with switch and medial then it is also derivable with switch, medial, and units. The following proposition, while not necessary for the proof of \cref{cor:main-red-to-main}, allows the the converse in some cases, and is the reason why our search algorithm in \cref{sec:algorithm} will only check for $\{\s,\m\}$-derivability.

\begin{prop}
[Follows from \cite{Das13:lin-inf-rew}, Lemma 28]
\label[prop]{non-triv-const-free-neg-free-derivability-without-units}
  Suppose \(\phi \to \psi \) is a non-trivial constant-free negation-free linear inference that is $\{\s,\m \}$-derivable with units. Then \(\phi \to \psi\) is also $\{\s,\m \}$-derivable (without units).
\end{prop}
\begin{proof}
    Suppose we have the following derivation of $\phi \to \psi$, a non-trivial constant-free negation-free linear inference:
    \begin{equation}
        \label{eq:nontriv-msu-derivation}
        \phi \sim_\acu \phi_0 \to_\ms \psi_0 \dots \phi_n \to_\ms \psi_n \sim_\acu \psi
    \end{equation}
    Using \Cref{prop:unit-free}, we can reduce any formula $\phi$ to a $\sim_\un$ equivalent constant-free formula (or $\bot$ or $\top$), which we will denote $\phi'$.
    We can now reduce each $\phi_i$ and $\psi_i$ to $\phi'_i$ or $\psi_i'$ respectively.
    Now each $\sim_\acu$ can be replaced by a $\sim_\ac$ by \Cref{acu-and-logical-equivalence} as all the involved formulae are constant-free or \(\bot\) or \(\top\).
    It remains to show that each rewrite is still valid.

    We prove by induction on the definition of $\to_\ms$ (namely as the closure of $\ms$ under contexts and substitution) that if $\phi \to_\ms \psi$ then either $\phi' \sim_\ac \psi'$ or $\phi' \rightsquigarrow_\ms \psi'$.
    \begin{itemize}
        \item Suppose we have an instance of the form $\phi \land \chi \to \psi \land \chi$, with $\phi \to_\ms \psi$.
        \begin{itemize}
            \item If $\chi' = \bot$ then indeed $(\phi \land \chi)' = (\psi \land \chi)' = \bot$.
            \item If $\chi' = \top$ then $(\phi \land \chi)' = \phi'$ and $(\psi\land \chi)' = \psi'$, whence we conclude by the inductive hypothesis for $\phi \to_\ms \psi$.
            \item Otherwise we have $\chi'$ is not a constant, and so $\phi'$ and $\psi'$ must also not be constants, or the inference would be trivial. Therefore, $(\phi \land \chi)' = \phi' \land \chi'$ and $(\psi \land \chi)' = \psi' \land \chi' $.
            By the inductive hypothesis for $\phi \to \psi$ we have $\phi' \sim_\ac \psi'$ or $\phi' \rightsquigarrow_\ms \psi'$, whence we conclude by context closure of $\sim_\ac $ and $\rightsquigarrow_\ms$.
        \end{itemize}
        \item The other cases for closure under contexts follow similarly.
        \item Now suppose $\phi \to_\ms \psi$ with $\phi = \alpha (\chi_0, \dots, \chi_n)$ and $\psi = \beta (\chi_0, \dots, \chi_n)$ where the inference $\alpha(x_0, \dots, x_n) \to \beta(x_0, \dots, x_n)$ (all variables displayed) is switch or medial (and so $n<4$).
        \begin{itemize}
            \item If each $\chi_i'$ is not a constant $\bot $ or $\top$, then also $\phi' = \alpha(\chi_0', \dots, \chi_n')$ and $\psi' = \beta (\chi_0', \dots, \chi_n')$ so $\phi' \to \psi'$ is again an instance of switch or medial, respectively.
            \item Otherwise, $\phi' \to \psi'$ is an instance of some (non-trivial) inference $\mathsf r$ on $<4$ variables.
            However any such $\mathsf r$ is either an instance of $\sim_\ac$ or $\rightsquigarrow_\s$. (This can be checked by hand or using our computer implementation presented in \Cref{sec:algorithm}).
        \end{itemize}
    \end{itemize}
    Now, returning to \cref{eq:nontriv-msu-derivation}, we have by the above argument that each $\phi_i' \sim_\ac \psi_i'$ or $\phi_i' \rightsquigarrow_\ms \psi_i'$.
    Since $\rightsquigarrow_\ms \ = \ \sim_\ac \cdot \to_\ms \cdot \sim_\ac$, we can locally merge any redundant $\ac$ steps to obtain that, indeed $\phi' \red_\ms \psi'$.
\end{proof}

\subsection{Minimality of inferences}
Let us take a moment to explain the various notions of `inference minimality' that we shall mention in this work.

\textbf{Size minimality} refers simply to the number of variables the inference contains. E.g.\ when we say that the \(8\)-variable inferences in the next section are size minimal (or `smallest') non-\(\{\s,\m\}\)-derivable with units (or \textbf{$\{\s,\m \}$-independent}) linear inferences, we mean that there are no \(\{\s,\m\}\)-independent linear inferences with fewer variables.

A linear inference \(\phi \to \psi\) is \textbf{logically minimal} if there is no logically distinct interpolating linear formula.
  I.e.\ if \(\phi \to \chi\) and \(\chi \to \psi\) are linear inferences, then $\chi$ is logically equivalent to $\phi$ or $\psi$ (and so, by \cref{acu-and-logical-equivalence}, is $\sim_\acu$-equivalent to \(\phi\) or \(\psi\)).

  Finally, a linear inference \(\phi \to \psi\) is \textbf{\(\{\s,\m\}\)-minimal} if there is no formula $\chi$ s.t.\ \(\phi \rightsquigarrow_\ms \chi\) or \(\chi \rightsquigarrow_\ms \psi\) and $\chi \to \psi$ or $\phi \to \chi$, respectively, is a valid linear inference which is not a logical equivalence.

It is clear from the definitions that any logically minimal inference is also \(\{\s,\m\}\)-minimal, though the converse may not be true.
The reason for considering \(\{\s,\m\}\)-minimality is that it is easier to systematically check by hand.
In fact, the implementation we give later in \cref{sec:algorithm} further verifies that our new $8$-variable inferences are logically minimal.

Logical minimality also serves an important purpose for our proof of \cref{thm:main-reduced}, as it allows the following reduction, greatly reducing the search space for our implementation, in fact to nearly $1\%$ of its original size for \(8\) variable inferences:\footnote{Reduces from \(514486\) inferences to \(5364\) and \(\nicefrac{5364}{514486} \approx 1.04\%\).}

\begin{lem}\label[lem]{lem:minimality}
  Suppose the statement of \cref{thm:main-reduced} holds whenever \(\phi \to \psi\) is logically minimal. Then the statement of \cref{thm:main-reduced} holds (even when \(\phi \to \psi\) is not logically minimal).
\end{lem}
\begin{proof}
  Suppose we have a non-trivial inference between constant-free negation-free linear inferences \(\phi \to \psi\). Then \(\phi \to \psi\) can be refined into a chain of logically minimal linear inferences \(\phi \to \chi_0 \to \dots \to \chi_n \to \psi\). All of these must be non-trivial, as triviality of any of them would imply triviality of \(\phi \to \psi\), cf.~\cref{trivial-composition}. Therefore if all such inferences are derivable from switch and medial (with units) then so is \(\phi \to \psi\), by transitivity.
\end{proof}

\section{New 8-variable \texorpdfstring{$\{\s,\m\}$}{\{s,m\}}-independent linear inferences}
\label{sec:8var-inf}
In this section we shall present the new 8-variable linear inferences of this work (\eqref{eq:php32-derived-inf}, \eqref{eq:unknown-inference-8}, and \eqref{eq:counterexample-inference} from the introduction), and give self-contained arguments for their $\{\s,\m\}$-independence and $\{\s,\m\}$-minimality, as a sort of sanity check for the implementation described in the next section.
We shall also briefly discuss some of their structural properties, in reference to previous works in the area.
Thanks to the results of the previous section, in particular \cref{prop:const-free-neg-free} and \cref{trivial-composition}, we shall only consider non-trivial constant-free negation-free linear inferences with the same variables in the LHS and RHS.
Furthermore, by \cref{non-triv-const-free-neg-free-derivability-without-units} we shall only consider $\{\s,\m \}$-derivability (i.e. without units).

The content of this section is based on Section 3 of the preliminary version, \cite{DasRice21}, with the addition of the inference \eqref{eq:unknown-inference-8} and surrounding discussion that was previously overlooked.

\subsection{Previous linear inferences}
\label{sec:prev-lin-infs}
In \cite{Str12:ext-wo-cut} Stra{\ss}burger presented a 36-variable inference that is $\{\s,\m\}$-independent, by an encoding of the pigeonhole principle with 4 pigeons and 3 holes.
He there referred to it as a `balanced' tautology, but in our setting it is a linear inference that can be written as follows:\footnote{We write Stra{\ss}burger's inference by encoding each $q_{i1j}$ as $\pp i j $, each $q_{i2j}$ as $\qq i j $, each $q_{i3j}$ as $\rr i j $, and using `primed' variables instead of duals, with the LHS of the inference being the appropriate instances of excluded middle.}
\[
\begin{array}{rc}
& \bigwedge\limits_{i=1}^3 \bigwedge\limits_{j=1}^i (\pp i j \lor \pp i j ' )
\land
\bigwedge\limits_{i=1}^3 \bigwedge\limits_{j=1}^i (\qq i j \lor \qq i j ' )
\land
\bigwedge\limits_{i=1}^3 \bigwedge\limits_{j=1}^i (\rr i j \lor \rr i j ' )
\\
\noalign{\smallskip}
\to &
\left[
\begin{array}{r@{\quad ( (}l@{)\ \land \ (}l@{)\ \land \ (}l@{))}}
	 & \pp 1 1  \lor \pp 2 1  \lor \pp 3 1   & \qq 1 1 \lor \qq 2 1 \lor \qq 3 1  & \rr 1 1 \lor \rr 2 1 \lor \rr 3 1  \\
\lor  & \pp 1 1 ' \lor \pp 2 2 \lor \pp 3 2 & \qq 1 1 ' \lor \qq 2 2 \lor \qq 3 2 & \rr 1 1 ' \lor \rr 2 2 \lor \rr 3 2 \\
\lor & \pp 2 1 ' \lor \pp 2 2 ' \lor \pp 3 3 & \qq 2 1 ' \lor \qq 2 2 ' \lor \qq 3 3 & \rr 2 1 ' \lor \rr 2 2 ' \lor \rr 3 3 \\
\lor  & \pp 3 1 ' \lor \pp 3 2 ' \lor \pp 3 3 ' & \qq 3 1 ' \lor \qq 3 2 ' \lor \qq 3 3 ' & \rr 3 1 ' \lor \rr 3 2 ' \lor \rr 3 3 '
\end{array}
\right]
\end{array}
\]
In \cite{Das13:lin-inf-rew} Das noticed that a more succinct encoding of the pigeonhole principle could be carried out, with only 3 pigeons and 2 holes, resulting in a 10-variable $\{\s,\m \}$-independent linear inference.
A variation of that, e.g.\ as used in \cite{Das17:unavoidable-con-loop}, is the following:
\begin{equation}
\label{eq:10varinf-nonminimal}
\begin{alignedat}{2}
&&& (z \lor (w \land w')) \land (y \lor y') \land (u \lor u') \land ((x \land x') \lor z') \\
&\to &\quad& (z \land (x \lor y)) \lor (u \land x') \lor (w' \land u') \lor ((w \lor y') \land z')
\end{alignedat}
\end{equation}
In fact this is not a $\{ \s,\m \}$-minimal inference, but we write this one here for comparison to two of the new 8-variable inferences in the next subsection.
It can be checked valid and non-trivial by simply checking all cases, or by use of a solver.
We do not give an argument for $\{\s,\m \}$-independence here, but such an argument is similar to the one we give for an 8-variable inference \cref{eq:php32-derived-inf-repeated} given in the next subsection.

\begin{rem}
[Deriving 3-2-pigeonhole using contraction]
  It is known that, when we add to $\{\s,\m\}$ structural rules that allow us to `duplicate' and `erase' propositional variables then the resulting (non-linear) system can derive every linear inference.
  This is thanks to the \emph{completeness} of the deep inference system $\mathsf{SKS}$ for propositional logic and its cut-elimination theorem (see, e.g., \cite{Brunnler03:thesis}).
  With respect to \cref{eq:10varinf-nonminimal} above, we need only a single `duplication' on $z$ or $z'$.
  For instance, here is the corresponding derivation for $z$,
  \[
  \begin{array}{ll}
       & \underline{((\purple z \land \orange z) \lor (w \land w'))} \land (y \lor y') \land (u \lor u') \land ((x \land x') \lor z')
       \\
       \to_\m & (\purple z \lor w) \land (\orange z \lor w') \land (y \lor y') \land (u \lor u') \land ((x \land x') \lor z')\\
       \sim_\ac & \underline{(\purple z \lor w) \land (y \lor y')} \land \underline{(\orange z \lor w') \land (u \lor u')} \land (z' \lor (x \land x'))\\
       \red_\s & ((\purple z \land y) \lor w \lor y') \land ( \orange z \lor u \lor (w' \land u') ) \land (z' \lor (x \land x')) \\
       \sim_\ac & \underline{((\purple z \land y) \lor w \lor y') \land (z' \lor (x \land x'))} \land (\orange z \lor u \lor (w' \land u')) \\
       \red_\s & ((\purple z \land y) \lor ((w \lor y') \land z') \lor (x \land x')) \land ((\orange z \lor u) \lor (w'\land u')) \\
       \red_\s & (\purple z \land y ) \lor ((w \lor y') \land z') \lor (\underline{((x \land x')\land (\orange z \lor u))} \lor (w' \land u')) \\
       \red_\s & (\purple z \land y ) \lor ((w \lor y') \land z') \lor (x \land \orange z) \lor (x' \land u) \lor (w' \land u') \\
       \sim_\ac & \underline{(\purple z \land x) \lor (\orange z \land y)}\lor (u \land x') \lor (w' \land u') \lor ((w \lor y') \land z')\\
       \rightsquigarrow_\m & ((\purple z \lor \orange z) \land (x \lor y) ) \lor (u \land x') \lor (w' \land u') \lor ((w \lor y') \land z')
  \end{array}
  \]
  where we have underlined some redexes for clarity.
  The purple $\purple z$ and orange $\orange z$ form what is known as a \emph{contraction loop} in the deep inference literature, cf.~\cite{Das15}.
\end{rem}

\subsection{The four minimal 8 variable \texorpdfstring{$\{\s,\m\}$}{\{s,m\}}-independent linear inferences}
\label{sec:two-found-inferences}
Pre-empting \cref{sec:search}, let us explicitly give the four minimal linear inferences found by our algorithm and justify their $\{\s,\m \}$-independence and $\{\s,\m \}$-minimality, as a sort of sanity check for our implementation later.
As we will see, they turn out to be significant in their own right, which is why we take the time to consider them separately.

\subsubsection{Refinements of the 3-2-pigeonhole-principle}
\label{sec:php32-refined}
First let us consider an 8 variable linear inference that may be used to derive \cref{eq:10varinf-nonminimal} (identical to \eqref{eq:php32-derived-inf} from the introduction):%
\begin{equation}
\label{eq:php32-derived-inf-repeated}
\begin{alignedat}{2}
&&& (z \lor (w \land w')) \land ((x \land x') \lor ((y \lor y') \land z')) \\
&\to &\quad& (z \land (x \lor y)) \lor ((w \lor y') \land ((w' \land x') \lor z'))
\end{alignedat}
\end{equation}

Recalling the notion of `duality' from \cref{example-weakening-duality}, let us formally define the \textbf{dual} of a linear inference $\phi \to \chi$ to be the linear inference $\bar \chi \to \bar \phi$, where $\bar \phi$ and $\bar \chi$ are obtained from $\phi$ and $\chi$, respectively, by flipping all $\lor$s to $\land$s and vice-versa.
Considering linear inferences up to renaming of variables, we have:

\begin{obs}
\label[obs]{duality-of-php32-8var}
\eqref{eq:php32-derived-inf-repeated} is self-dual.
\end{obs}
\noindent
Indeed, the formula structure of the RHS is clearly the dual of that of the LHS, and the mapping from a variable in the LHS to the variable at the same position in the RHS is, in fact, an involution. I.e., $u$ is mapped to itself; $v$ is mapped to $y$ which in turn is mapped to $v$; $v'$ is mapped to $w$ which is in turn mapped to $v'$; $x $ is mapped to $y'$ which in turn is mapped to $x$; and $z$ is mapped to itself.
Validity may be routinely checked by any solver, but we give a case analyis of assignments in \cref{sec:app:validity-php32-derived}.

We may also establish $\{\s,\m \}$-independence and $\{\s,\m \}$-minimality by checking all applications of $\s$ or $\m$ to the LHS (note that we do not need to check the RHS, by \cref{duality-of-php32-8var} above).
This analysis is given explicitly in \cref{sec:ind-min-php32}.

We also found the following inference (identical to \cref{eq:unknown-inference-8}):

\begin{equation}
    \label{eq:unknown-inference-8-repeated}
    \begin{alignedat}{2}
    & &&(z \lor (w \land w')) \land ((x \land x') \lor ((y \lor y') \land z'))\\
    & \to &\quad&((z \lor (w \land x)) \land (x' \lor y)) \lor ( (w' \lor y') \land z')
    \end{alignedat}
\end{equation}
At first glance this looks very similar to \eqref{eq:php32-derived-inf-repeated}, however it turns out that it is not isomorphic to it, and neither is derivable from the other. Like \eqref{eq:php32-derived-inf-repeated}, it is also self dual. Validity of this inference is checked in \Cref{sec:app:validity-unknown-inf}, and its $\{\s,\m\}$-minimality is checked in \Cref{sec:ind-min-unknown} (similarly to the \eqref{eq:php32-derived-inf-repeated}, by duality we only need to check applications of \(\s\) and \(\m\) from the LHS).

\begin{exa}
[Deriving 3-2-pigeonhole using \eqref{eq:php32-derived-inf-repeated} and \eqref{eq:unknown-inference-8-repeated}]
As we suggested, the linear inference \cref{eq:10varinf-nonminimal} is in fact derivable from each of \cref{eq:php32-derived-inf-repeated} and \cref{eq:unknown-inference-8-repeated} above
 (modulo $\{\s\}$).
Here is the derivation in $\{\s,\eqref{eq:php32-derived-inf-repeated} \}$,
\[
\begin{array}{ll}
    &(z \lor (w\land w')) \land (y \lor y') \land (u \lor u') \land ((x \land x') \lor z') \\
     \sim_\ac & (z \lor (w\land w')) \land \underline{(y \lor y') \land ((x \land x') \lor z')} \land (u \lor u') \\
     \rightsquigarrow_\s & \underline{(z \lor (w \land w')) \land ((x \land x') \lor ((y \lor y') \land z'))} \land (u \lor u') \\
     \to_{\eqref{eq:php32-derived-inf-repeated}} & ((z \land (x \lor y)) \lor ( \underline{(w \lor y') \land ((w'\land x') \lor z')) }) \land (u \lor u') \\
     \rightsquigarrow_\s & ((z \land (x \lor y)) \lor (w' \land x') \lor ((w \lor y') \land z') ) \land (u \lor u') \\
     \red_\s & ((z \land (x \lor y)) \lor (\underline{w' \land x' \land (u \lor u')}) \lor ((w \lor y') \land z') ) \\
     \red_\s & (z \land (x \lor y)) \lor (u \land x') \lor (w' \land u') \lor ((w \lor y') \land z')
\end{array}
\]
where we have underlined certain redexes.
In $\{\s,\eqref{eq:unknown-inference-8-repeated}\}$ we have the derivation:
\[
\begin{array}{ll}
    &(z \lor (w\land w')) \land (y \lor y') \land (u \lor u') \land ((x \land x') \lor z') \\
     \sim_\ac & (z \lor (w\land w')) \land \underline{(y \lor y') \land ((x \land x') \lor z')} \land (u \lor u') \\
     \rightsquigarrow_\s & \underline{(z \lor (w \land w')) \land ((x \land x') \lor ((y \lor y') \land z'))} \land (u \lor u') \\
     \to_{\eqref{eq:unknown-inference-8-repeated}} & (\underline{((z \lor (w' \land x')) \land (x \lor y))} \lor ((w \lor y')\land z') ) \land (u \lor u') \\
     \rightsquigarrow_\s & ((z \land (x \lor y)) \lor (w' \land x') \lor ((w \lor y') \land z') ) \land (u \lor u') \\
     \red_\s & ((z \land (x \lor y)) \lor (\underline{w' \land x' \land (u \lor u')}) \lor ((w \lor y') \land z') ) \\
     \red_\s & (z \land (x \lor y)) \lor (u \land x') \lor (w' \land u') \lor ((w \lor y') \land z')
\end{array}
\]
Again we have underlined some redexes and, in the instance of \cref{eq:unknown-inference-8-repeated}, we have applied the permutation $\{(w,w'),(x,x')\} $.
Note that these two derivations are almost identical, differing only in the switch application immediately after the instance of \cref{eq:php32-derived-inf-repeated} or \cref{eq:unknown-inference-8-repeated} respectively.
\end{exa}

\subsubsection{Counterexamples to a conjecture of Das and Stra{\ss}burger}
\label{sec:counterexample-inference}
Finally, our search algorithm found two completely new linear inferences, which were dual to eachother. One of them
(identical to \eqref{eq:counterexample-inference} from the introduction) is:
\begin{equation}
  \label{eq:counterexample-inference-repeated}
  \begin{alignedat}{2}
    & &&((w \land w') \lor (x \land x')) \land ((y \land y') \lor (z \land z'))\\
    & \to&\quad& (w \land y) \lor ((x \lor (w'\land z')) \land ((x'\land y') \lor z) )
  \end{alignedat}
\end{equation}
with the other given by dualising the above inference.
Again, validity is routine, but a case analysis is given in \cref{sec:app:validity-counterexample-inference}.
We may establish $\{\s,\m \}$-independence and $\{\s,\m \}$-minimality again by checking all possible rule applications.
This analysis is given in \cref{sec:ind-min-counterexample-inference}.

This new inference exhibits a rather interesting property, which we shall frame in terms of the following notion, since it will be used in the next section:
\begin{defi}
\label[defi]{definition-of-lccs}
Let $\phi$ be a linear formula on a variable set $\V$. For distinct $x,y \in \V$, the \textbf{least common connective} (lcc) of $x$ and $y $ in $\phi$ is the connective $\lor$ or $\land$ at the root of the smallest subformula of $\phi$ containing both $x$ and $y$.
\end{defi}
Note that, in the inference \eqref{eq:counterexample-inference-repeated} above, the lcc of $w' $ and $x'$ changes from $\lor$ to $\land$, but the lcc of $y$ and $y'$ changes from $\land$ to $\lor$.
No such example of a minimal linear inference exhibiting both of these properties was known before; switch, medial and all of the linear inferences of this section either preserve $\lor$-lccs or preserve $\land$-lccs.
In fact, Das and Stra{\ss}burger showed that any valid linear inference preserving $\land$-lccs is already derivable by medial \cite[Theorem 7.5]{DasStr16:no-compl-lin-sys}.
Furthermore they made the following conjecture:

\medskip

\noindent
\textit{Conjecture 7.9 from \cite{DasStr16:no-compl-lin-sys}, rephrased.}
Each logically minimal nontrivial linear inference preserves either $\land$-lccs or $\lor$-lccs.

\medskip

\noindent
Naturally, \eqref{eq:counterexample-inference-repeated} and $\overline{\eqref{eq:counterexample-inference-repeated}}$ are counterexamples to that:

\begin{thm}
Conjecture 7.9 from \cite{DasStr16:no-compl-lin-sys} is false.
\end{thm}

We shall see several further counterexamples over 9 variables later in \Cref{sec:9var-infs}, where we also discuss possible consequences and trends arising from all these new inferences.

\section{A graph-theoretic presentation of linear inferences}
\label{sec:webs}
In this section we consider graphical representations of formulas that render our search algorithm in \Cref{sec:algorithm} more efficient.
The content of this section is based on Section 4 from the preliminary version, \cite{DasRice21}, with the addition of further examples, figures and narrative.

A significant cause of algorithmic complexity when searching for linear inferences is the multitude of formulae equivalent modulo associativity and commutativity ($\sim_\ac$). For example, for 7 variables, there are \(42577920\) formulae (ignoring units), yet only \(78416\) equivalence classes.
Under \cref{acu-and-logical-equivalence} it would be ideal if we could deal with $\sim_\ac$-equivalence classes directly, realising logical and syntactic notions on them in a natural way.
This is precisely what is accomplished by the graph-theoretic notion of a \emph{relation web}, cf.~ \cite{Gug07:sys-int-struct,Str07:char-med,DasStr15:no-comp-lin-sys,DasStr16:no-compl-lin-sys}.

Throughout this section we work only with constant-free negation-free linear formulae, cf.~\cref{thm:main-reduced}.
Recall the notion of \emph{least common connective} (lcc) from \cref{definition-of-lccs}.

\begin{defi}
Let $\phi$ be a linear formula on a variable set $\V$.
The \textbf{relation web} (or simply \textbf{web}) of $\phi$, written $\W(\phi)$, is a simple undirected graph with:
\begin{itemize}
\item The set of nodes of $\W (\phi)$ is just $\V$, i.e.\ the variables occurring in $\phi$.
\item For $x,y \in \V$, there is an edge between $x$ and $y$ in $\W(\phi)$ if the lcc of $x$ and $y$ in $\phi$ is $\land$.
\end{itemize}
\end{defi}

When we draw graphs, we will draw a solid red line \(\redge[]{x}{y}\) if there is an edge between \(x\) and \(y\), and a green dotted line \(\gedge[]{x}{y}\) otherwise.
\begin{exa}\label[exa]{ex:relation-web}
  Let $\phi$ be the linear formula \(w \land (x \land (y \lor z))\). \(\W(\phi)\) is the following graph:

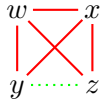
\begin{center}
    $ \FourGraph{w,x,y,z}rrrrrg $
\end{center}
\end{exa}
  Note that linear formulae equivalent up to associativity and commutativity have the same relation web, since $\sim_\ac$ does not affect the lccs.
  For instance, if \(\psi = (w \land x) \land (z \lor y)\), then $\W(\psi)$ is still just the relation web above.
  In fact, we also have the converse:
  \begin{prop}
  [E.g., \cite{DasStr16:no-compl-lin-sys}, Proposition 3.5]
  \label[prop]{prop:webs-equiv-classes-ac}
    Given linear formulae \(\phi\) and \(\psi\), \(\phi \sim_\ac \psi\) if and only if \(\W(\phi) = \W(\psi)\).
  \end{prop}
Thus relation webs are natural representations of equivalence classes of linear formulae modulo associativity and commutativity.

We can also realise our new 8-variable inferences from the previous section graphically:
\begin{exa}
    The inferences \eqref{eq:php32-derived-inf-repeated}, \eqref{eq:unknown-inference-8-repeated} and \eqref{eq:counterexample-inference-repeated} may be construed as graph rewrites on their webs, as given in \cref{fig:8-var-infs-webs}.
    The one for $\overline{\eqref{eq:counterexample-inference-repeated}}$ is obtained from that of \eqref{eq:counterexample-inference-repeated} by switching every edge to a non-edge and vice-versa, and exchanging the LHS and RHS.
Note, for \eqref{eq:counterexample-inference-repeated}, that the edge between $y'$ and $z'$ changes from red to green, and several edges (e.g. $x'$ to $y$) change from green to red, exhibiting how this inference breaks Conjecture 7.9 from \cite{DasStr16:no-compl-lin-sys}.
\end{exa}

\begin{figure}
    \centering
    \begin{equation*}
  \resizebox{0.4\hsize}{!}{\(\EightGraph{w,w',x,x',y,y',z,z'}rrrrrgrrrrrgrrggrgggrggrrrrr\)}
  \quad\raisebox{29pt}{\(\to\)}\quad
  \resizebox{0.4\hsize}{!}{\(\EightGraph{w,w',x,x',y,y',z,z'}rgrgggrgrgrgggggrggrgggrggrg\)}
\end{equation*}
\begin{equation*}
  \resizebox{0.4\hsize}{!}{\(\EightGraph{w,w',x,x',y,y',z,z'}rrrrrgrrrrrgrrggrgggrggrrrrr\)}
  \quad\raisebox{29pt}{\(\to\)}\quad
  \resizebox{0.4\hsize}{!}{\(\EightGraph{w,w',x,x',y,y',z,z'}grrrggggggggrrrgggggrggrggrg\)}
\end{equation*}
    \begin{equation*}
  \resizebox{0.4\hsize}{!}{\(\EightGraph{w,w',x,x',y,y',z,z'}rggrrrrggrrrrrrrrrrrrrrggrgr\)}
  \quad\raisebox{29pt}{\(\to\)}\quad
  \resizebox{0.4\hsize}{!}{\(\EightGraph{w,w',x,x',y,y',z,z'}gggrggggrgrrrrgrrggrgrggggrr\)}
\end{equation*}
    \caption{Graphical versions of the 8-variable inferences \eqref{eq:php32-derived-inf-repeated}, \eqref{eq:unknown-inference-8-repeated} and \eqref{eq:counterexample-inference-repeated} respectively.}
    \label{fig:8-var-infs-webs}
\end{figure}

It is easy to see that the image of $\W$ is just the \emph{cographs}. A \textbf{cograph} is either a single node, or has the form \(\redge \R \S\) or $\gedge \R \S$ for cographs $\R$ and $\S$.
Formally, $\redge \R \S$ has as nodes the disjoint union of the nodes of $\R$ and the nodes of $\S$; edges within the $\R$ component are inherited from $\R$ and similarly for $\S$; there is also an edge between every node in $\R$ and every node in $\S$. $\gedge \R \S$ is defined similarly, but without the last clause.

A \textbf{cograph decomposition} of a cograph $\R$ is just a definition tree according to these construction rules (its `cotree'), which we identify with a linear formula (whose web is $\R$) in the natural way, setting $\land$ for a $\redge{}{}$ node and $\lor$ for a $\gedge{}{}$ node.
By \cref{prop:webs-equiv-classes-ac}, the cograph decomposition of a cograph is unique up to associativity and commutativity of $\land $ and $\lor$.

Cographs admit an elegant \emph{local} characterisation by means of forbidden subgraphs:
\begin{defi}
An \textbf{induced} subgraph of a graph $G$ is one whose edges are just those of $G$ restricted to some subset of the nodes.
  A graph \(G\) is \textbf{\(P_4\)-free} if none of its induced subgraphs are isomorphic to the following graph, called \(P_4\):
  \begin{center}
    \(\FourGraph{w,x,y,z}rgggrr\)
  \end{center}
  Further define a \textbf{clique} \(P\) of a graph \(G\) to be a subset of the nodes of \(G\) such that the induced subgraph on these nodes has all possible edges. Dually define a \textbf{stable set} to be a subset of the nodes whose induced subgraph contains no edges. A clique or stable set is \textbf{maximal} if any proper superset is no longer a clique or stable set respectively.
\end{defi}

The following result is well-known since (at least) the '70s, and has appeared numerous times in the literature (see also \cite{Gug07:sys-int-struct,Str07:char-med}):
\begin{prop}
[E.g.\ \cite{CorneilLB81}]
  A graph is a cograph if and only if it is \(P_4\)-free.
\end{prop}

  Thus, relation webs are just the $P_4$-free graphs whose nodes are variables.
Note, in particular, that this characterisation gives us an easy way to check whether a graph is the web of some formula: just check every $4$-tuple of nodes for a $P_4$.

\begin{exa}
\label[exa]{ex:p5c5}
$P_5$ and $C_5$ are two examples of graphs which are not $P_4$ free. They are the following:
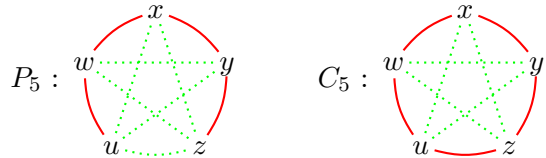
\begin{center}
    $P_5$\ :\ \(\FiveGraphCirc{u,w,x,y,z}rgggrggrgr\) \qquad $C_5$\ :\  \(\FiveGraphCirc{u,w,x,y,z}rggrrggrgr\)
\end{center}
In both these cases, the induced subgraph generated from \(\{u,w,x,y\}\) is $P_4$.
\end{exa}

What is more, we may also verify several semantic properties of linear inferences, such as validity and triviality, directly at the level of relation webs:

\begin{prop}
[Follows from Proposition 4.4 and Theorem 4.6 in \cite{DasStr16:no-compl-lin-sys}]
  \label[prop]{lem:rw-inference}
  Let $\phi$ and $\psi$ be linear formulae on the same set of variables with webs $\R$ and $\S$ respectively. The following are equivalent:
  \begin{enumerate}
      \item\label{item:valid-char-formulas} \(\phi \to \psi\) is a valid linear inference.
      \item\label{item:valid-char-cliques} For every maximal clique \(P\) of \(\R\), there is some \(Q \subseteq P\) such that \(Q\) is a maximal clique of \(\S\).
      \item\label{item:valid-char-stable-sets} For every maximal stable set \(Q\) of \(\S\), there is some \(P \subseteq Q\) such that \(P\) is a maximal stable set of \(\R\).
  \end{enumerate}
\end{prop}

Note that for inferences between arbitrary graphs (i.e. graphs that are not $P_4$-free), conditions \eqref{item:valid-char-cliques} and \eqref{item:valid-char-stable-sets} above are not equivalent.
In fact, (3) holds for $\R \to \S$ exactly when (2) holds for $\overline{\S} \to \overline{\R}$, where $\overline{\R}$ is the dual of $\R$ (inverting every edge).

\begin{exaC}
[\cite{DasStr16:no-compl-lin-sys}]
    Recall the definition of $P_5$ and $C_5$ from \Cref{ex:p5c5}. It can be checked that $P_5 \to C_5$ satisfies \eqref{item:valid-char-cliques} but $C_5 \to P_5$ does not (as $\{u,z\}$ is a maximal clique of $C_5$ with no subset that is a maximal clique of $P_5$). Conversely $C_5 \to P_5$ \emph{does} satisfy \eqref{item:valid-char-stable-sets} but $P_5 \to C_5$ does not (for example $\{x,z\}$ is a maximal stable set of $C_5$ but none of $\{x,z\}$ is stable but not maximal in $P_5$ as $\{u,x,z\}$ is stable).
\end{exaC}

Because $\W(\overline{\phi}) = \overline{\W(\phi)}$, we see that \eqref{item:valid-char-cliques} and \eqref{item:valid-char-stable-sets} coincide for $P_4$-free graphs (as $\phi \to \chi$ is a valid linear inference precisely when $\overline{\chi} \to \overline{\phi}$ is).
Such relationships are further discussed in \cite{CalDasWar20:bgl}.

\begin{propC}
[{\cite[Proposition~5.7]{DasStr16:no-compl-lin-sys}}]
  \label[prop]{lem:rw-trivial}
  Let $\phi$ and $\psi$ be linear formulae on the same variables with webs $\R$ and $\S$ respectively.
  The following are equivalent:
  \begin{enumerate}
      \item \(\phi \to \psi\) is a linear inference that is trivial at \(x\).
      \item For every maximal clique \(P\) of \(\R\), there is some \(Q \subseteq P \setminus \{x\}\) such that \(Q\) is a maximal clique of \(\S\).
      \item For every maximal stable set $Q$ of $\S$ there is some $P \subseteq Q \setminus \{x\}$ such that $P$ is a maximal stable set of $\R$.
  \end{enumerate}
\end{propC}
\noindent
Note that the criterion for triviality is a strict strengthening of that for validity, as we would expect.
Again note the duality between the second and third characterisations of triviality above.

\begin{exa}
[Validity of switch and medial, triviality of mix]
The switch and medial inferences can be construed as the following graph rewrite rules on relation webs, respectively:
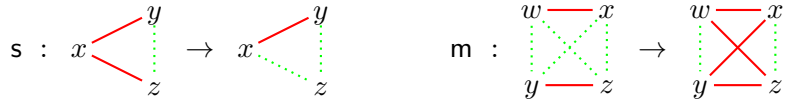
\begin{center}
\(
\s \ : \
\ThreeGraph{x,y,z}rrg
\ \to \
\ThreeGraph{x,y,z}rgg
\qquad\qquad
\m\ : \
\FourGraph{w,x,y,z}rggggr
\ \to \
\FourGraph{w,x,y,z}rgrrgr
\)
\end{center}
It is easy to see that the validity criterion of \cref{lem:rw-inference} holds for each of these rules.
For $\s$, the maximal cliques $\{x,y\}$ and $\{x,z\}$ in the LHS are mapped to $\{x,y\}$ and $\{z\}$ in the RHS respectively.
For $\m$, the maximal cliques $\{w,x\}$ and $\{y,z \}$ in the LHS are mapped to themselves in the RHS.

Now consider the trivial inference $x\land y \to x\lor y$, construed as the graph rewrite rule:
\begin{center}
$\redge x y \ \to \ \gedge x y $
\end{center}
We can easily verify the criterion for triviality at $x$ from \cref{lem:rw-trivial} since the only maximal clique on the LHS, $\{x,y\}$ has $\{y\} \subseteq \{x,y\}\setminus \{x\}$ as a maximal clique on the RHS.
\end{exa}

\begin{exa}
[Supermix]
A further interesting example is given by a class of inferences from \cite{Das13:lin-inf-rew} known as `supermix'. These have the following form:
\[ a \land \bigg(\bigvee_{i < n} b_i\bigg) \to a \lor \bigg(\bigwedge_{i < n} b_i\bigg) \]
Each of these inferences is trivial, as long as $n>0$, and moreover they are trivial at each $b_i$ `at the same time'.
Note that, for $n>3$, they are examples of inferences that that have more $\land$ symbols in the RHS than the LHS, and for $n>4$ more edges in the RHS than the LHS.
For instance, considering $n = 4$,
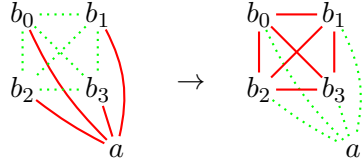
\[ \FiveGraph{b_0,b_1,b_2,b_3,a}gggrggrgrr \quad \to \quad \FiveGraph{b_0,b_1,b_2,b_3,a}rrrgrrgrgg\]
we can see that there are 4 edges in the LHS but $\binom 4 2 = 6$ edges in the RHS.
The only other inference so far that increases the number of edges is medial (which, nonetheless, reduces the number of $\land$ symbols in a formula).
It is not known whether any non-trivial linear inference has both of these properties.
\end{exa}

\begin{rem}\label[rem]{rem:using-webs}
  With the results in this section, the notation for inferences between formulae can be equally used for relation webs. For example, for webs \(\R\) and \(\S\), we write \(\R \rightsquigarrow_\ms \S\) (or $\R \to \S$ is an instance of $\{\s,\m\}$) to mean that the inference between (any choice of) the underlying linear formulae is an instance of $\rightsquigarrow_\ms$,
  and \(\R \redms \S\) to mean the there is a derivation from switch and medial between the underlying formulae.
  Since these relations are invariant under associativity and commutativity, they are independent of the particular cograph decomposition chosen.

  Furthermore, to prove \cref{thm:main-reduced}, it is sufficient to show that for all webs \(\R\) and \(\S\) with size less than \(8\), if \(\R \to \S\) is valid and non-trivial then \(\R \redms \S\).
\end{rem}

It is in fact possible to check whether an inference is derivable from medial at the level of graphs:

\begin{propC} [{\cite[Theorem~5]{Str07:char-med}}]\label[prop]{prop:medial-char}
Let $\R \to \S$ represent a constant-free negation-free non-trivial linear inference.
Then $\R \to \S$ is derivable from medial if and only if:
  \begin{itemize}
  \item Whenever \(\redge x y \) in $\R$, also \(\redge x y\)  in \( \S \).
  \item Whenever \(\gedge x y \) in \(\R \) but \(\redge x y \) in \( \S \) there exists \(w\) and \(z\) such that
    \begin{center}
      \(\FourGraph{w,x,y,z}rggggr\) is an induced subgraph of \(\R \) and \(\FourGraph{w,x,y,z}rgrrgr\) is an induced subgraph of \(\S \).
    \end{center}
  \end{itemize}
\end{propC}
The second condition can, in fact, be replaced by simply requiring that $\R \to \S$ is valid \cite[Theorem 7.5]{DasStr16:no-compl-lin-sys}.
A relation web characterisation for switch derivability can also be found in \cite[Theorem~6.2]{Str07:char-med}, however we do not use it in our implementation. The medial characterisation was used in the original implementation that accompanied \cite{DasRice21}, though a more general method is used in the current implementation which allows checking whether an inference is an instance of any graph rewrite (see \Cref{sec:library} for more details).

\section{Implementation}
\label{sec:algorithm}

As stated in previous sections, \cref{thm:main-reduced} is proved using a computational search.
In this section we describe the algorithm used to search for \(S\)-independent inferences between \(P_4\)-free graphs of size \(n\), for some set \(S\) of linear inferences, as well as some of the optimisations we employ so that this search finishes in a reasonable time.
Many of these optimisations may be of self-contained theoretical interest.
To prove \cref{thm:main-reduced}, we will take the set \(S = \{\s,\m\}\) and \(n = 7\), but everything in this section generalises to an arbitrary set of linear inferences.

The implementation is written in Rust,\footnote{\url{https://www.rust-lang.org/}} which offers a combination of good performance (both in terms of speed and memory management) but also provides a variety of high level abstractions such as algebraic data types. Furthermore, it has built-in support for iterators, allowing the code to be written in a more functional style, and has a built-in testing framework, meaning that sanity checks can be built into the code base. The code\cite{Ric21:implementation} is available at
\begin{center}
  \url{https://github.com/alexarice/lin_inf}
\end{center}
and has been split into two parts: a \emph{library} containing types for undirected linear graphs and formulae and some operations on them, and an \emph{executable} which implements the search algorithm using this library as a base.

The content in this section is based on Section 5 of the preliminary version \cite{DasRice21}.
In this version we have generalised the search algorithm to take an arbitrary set $S$ of linear inferences as input, and we have further used this to devise a recursive algorithm $\basis$ computing minimal sets of independent linear inferences, whose correctness we justify in \Cref{sec:further-theoretical}.
\subsection{Library}
\label{sec:library}

The library portion of the implementation defines methods for working with relation webs, as well as the ability to convert formulae to relation webs and vice versa. The majority of the library consists of the \texttt{LinGraph} trait, which is an interface for types that can be treated as undirected graphs. This allows us to query the edges between variables as well as perform more involved operations such as checking whether a graph is \(P_4\)-free.
We may also ask whether a pair of relation webs forms a valid linear inference and check whether the inference is trivial using \cref{lem:rw-inference,lem:rw-trivial}.

\smallskip

\textbf{Storing graphs and relation webs.}
The library was designed with the intention of storing graphs as compactly as possible.
 Therefore there are implementations of \texttt{LinGraph} which pack the data (a series of bits for whether there exists an edge between each pair of nodes) into various integer types. The implementation is given for unsigned 8 bit, 16 bit, 32 bit (which can store up to 8 variable graphs), 64 bit, and 128 bit integers. Furthermore there is an implementation using vectors (variable length arrays) of Boolean values, which is less memory efficient but can store relation webs of arbitrary size. A further improvement could be to use an external library implementing bit arrays to make a memory efficient, yet infinitely scalable implementation.

\smallskip

\textbf{Checking an inference between graphs.}
In order to implement linear inference checking, we use a data type representing maximal cliques of a relation web, which we represent as the trait \texttt{MClique}. It is possible to use Rust's inbuilt \texttt{HashSet}\footnote{\url{https://doc.rust-lang.org/std/collections/struct.HashSet.html}} to do this but, as above, a more memory efficient solution is provided where we store the data in a single integer, with each bit determining whether a node is contained in the clique. For example a maximal clique on an 8 node graph can be encoded into a single byte. While checking for linear inferences and triviality, the main operation on maximal cliques is asking whether one is a subset of the other. This operation can be carried out very quickly using bitwise operations. Lastly we also need a way to generate the maximal cliques of a relation web. This is done using the Bron-Kerbosch algorithm \cite{BroKer73:finding-all-maxcliques}, which is fast enough for our purposes (as we are only finding the maximal cliques of relatively small graphs).

As mentioned after \cref{lem:rw-inference}, for arbitrary (non-$P_4$-free graphs) the maximal clique condition and maximal stable set condition for checking an inference do not have to agree, and so we would like to have both options in our library. To avoid having to calculate stable sets as well as cliques, we instead simply calculate whether an inference $\R \to \S$ between graphs satisfies the stable set condition by checking whether $\overline{\S} \to \overline{\R}$ satisfies the clique condition. To this end, the library contains functions for dualising a graph.

\smallskip

\textbf{Working with isomorphism.}
There is also code for working with isomorphisms of graphs, which is used in the search algorithm to shrink the search space further. This is implemented as a module where permutations and operations on these permutations are defined, as well as having the ability to apply a permutation to the nodes of a graph, to get a new but isomorphic graph.

\smallskip

\textbf{Generating all $P_4$-free graphs.}
The library also has a function that allows all \(P_4\)-free graphs of a certain size to be generated. The naive algorithm for doing this which simply generates all graphs and checks each one for being \(P_4\)-free is computationally infeasible for graphs with more than a few variables, as the number of graphs scales superexponentially with the number of variables (for instance there are \(2^{21}\) 7-variable undirected graphs).
Instead, we use a recursive algorithm that generates all \(P_4\)-free graphs of size \(n\) by first generating the \(P_4\)-free graphs of size \(n-1\) and then checking all possible extensions of these graphs to see if they are \(P_4\)-free. Correctness of this procedure is due to the fact that induced subgraphs of a \(P_4\)-free are themselves \(P_4\)-free.
In fact, a further optimisation is also added: when we check whether the extensions are \(P_4\)-free, it is sufficient to only check if subsets of the nodes containing the added node are not isomorphic to \(P_4\), instead of checking every subset.

\smallskip

\textbf{Checking if an inference is an instance of another}.
The library contains code for checking whether an inference $\mathsf r : \R \to \R'$ is an instance of another inference $\mathsf p : \S \to \S'$.
This is implemented as in \cite{CalDasWar20:bgl}, and so allows us to also work with non-$P_4$-free graphs.
Let us first recall the definition of a module:
\begin{defi}
Given a graph $\R$ on variables $\V$, a subset $M$ of $\V$ is a \emph{module} if every element of $M$ has the same neighbourhood outside of $M$, so if $x, y\in M$ and $z \in \V \setminus M$ then there is an edge between $x$ and $z$ if and only if there is an edge between $y$ and $z$.
\end{defi}
For a formula $\phi$, the modules of the graph $\W(\phi)$ correspond exactly to the subformulae of $\phi$.
The check is done as follows:
\begin{itemize}
    \item For each subset $M$ of the variable set, check if $M$ is both a module of $\R$ and $\R'$. This is the generalisation of finding a suitable context for the rewrite to all graphs. Each such module \(M\) generates the induced subgraphs $\mathcal{M} \subseteq \R$ and $\mathcal{M'} \subseteq \R'$.
    \item If $\mathsf r$ is an instance of $\mathsf p$, then for some $M$ found we have there is a substitution of $\mathsf p$ which produces $\mathcal{M} \to \mathcal{M'}$. We check this by going through all partitions of $M$ into $|\S|$ non-empty parts $P_x$ for every node $x$ of $\S$ and seeing if this substitution produces $\mathcal{M} \to \mathcal{M'}$. This is done by checking if $\mathcal{M}$ has an edge between $a \in P_x$ and $b \in P_y$ (with $x \neq y$) precisely when then there is an edge between $x$ and $y$ in $\S$, checking the similar condition for $\mathcal{M}'$, and finally checking that for each $x$, the induced subgraph of $\mathcal{M}$ and $\mathcal{M}'$ generated from $P_x$ are the same.
\end{itemize}

When specialised to $P_4$-free graphs, this procedure agrees with the usual notion of rewrite step.

\smallskip

\textbf{Sanity checks.}
Finally, the library also contains some automated tests used as sanity checks on the code, which may be used to check various implementations against each other.

\subsection{Search algorithm}
\label{sec:search}

The main part of the implementation is a search algorithm to find logically minimal non-trivial inferences of a specified size \(n\) between relation webs that are not derivable from a specified set of inferences \(S\). The search algorithm functions in multiple phases. After each phase the results are serialised and saved to disk so that the algorithm can be restarted from this point.

\smallskip

\textbf{Phase 1: generating $P_4$-free graphs on $n$ nodes.}
Suppose we are searching for \(S\)-independent linear inferences between webs on \(n\) variables. The first phase, as described in the previous section, is to gather all \(P_4\)-free graphs with \(n\) nodes.

\smallskip

\textbf{Phase 2: identifying isomorphism classes and canonical representatives.}
To describe the second phase we need to introduce some new notions. Without loss of generality, we will assume henceforth that the variable set is given by \(\V = \{0,\dots,n-1\}\).

\begin{defi}\label[defi]{def:num-rep}
Note that the function $\iota : \{(x,y) \in \mathbb{N} \times \mathbb{N}\ |\ x < y\} \to \mathbb{N}$ given by \(\iota(x,y) = x + \sum_{i < y} i\) is a bijection.
  Define the \textbf{numerical representation} of a linear graph \(\R\), written \(N(\R)\), to be the natural number whose \(\iota(x,y)\)\textsuperscript{th} least significant bit is \(1\) if and only if \((x,y) \in \R\).
\end{defi}
\noindent
This is the encoding used to store graph in integers as described in the \Cref{sec:library}.

\begin{defi}
Given a bijection $\rho: \V\to \V$, we write $\rho(\R)$ for the graph on $\V$ with edges $(\rho(x),\rho(y))$ for each edge $(x,y) \in \R$.
$\R$ and $\S$ are \textbf{isomorphic} if $\S = \rho(\R)$, for some bijection $\rho: \V\to \V$, in which case $\rho$ is called an \textbf{isomorphism} from $\R $ to $\S$.
\end{defi}

As isomorphism is an equivalence relation, we can partition the set of \(P_4\)-free graphs into isomorphism classes. It can readily be checked that \(N\) (from \cref{def:num-rep}) is injective and can therefore be used to induce a strict total ordering on graphs. Say that a relation web is \textbf{least} if it is the smallest element in its isomorphism class (with respect to this ordering induced from \(N\)).

The second phase of the algorithm is to identify these least relation webs, as well as identify the isomorphism between every relation web and its isomorphic least relation web. It will become clear why this data is needed later on in the section. To obtain this, first the relation webs are sorted (by numerical representation) and then, taking each graph \(\R\) in turn, applying every possible permutation to its nodes, and seeing if any result in a smaller web (with respect to $N$). If none do then we record it as a least relation web (with the identity isomorphism). Otherwise suppose it is isomorphic to \(\R'\) with isomorphism \(\rho\) where \(N(\R') < N(\R)\). As we are checking graphs in order, we must already know that \(\R'\) is isomorphic to least graph \(\R''\) with isomorphism \(\pi\). Then we can record \(\R\) as being isomorphic to \(\R''\) with isomorphism \(\pi \circ \rho\). This allows us to use the following lemma.

\begin{lem}
  \label[lem]{lem:least}
  Take some set \(S\) of inferences and number \(n\). Then to show that any valid non-trivial logically minimal inference \(\R \to \S\) on \(n\) variables is derivable from \(S\), that is \(\R \red_{S} \S\), it is sufficient to only check those inferences where \(\R\) is least. Hence to prove the statement of \cref{thm:main-reduced}, it is sufficient to prove that for all valid non-trivial logically minimal inferences \(\R \to \S\) on \(n < 8\) variables where \(\R\) is least, that \(\R \redms \S\).
\end{lem}
\begin{proof}
  Suppose \(\R \to \S\) is a valid non-trivial logically minimal inference on \(n\) variables. Then let \(\rho\) be an isomorphism from \(\R\) to \(\R'\) which is least, and let \(\S' = \rho(S)\). Then we have that \(\R' \to \S'\) is valid, non-trivial, minimal, and \(\R'\) is least, and so by assumption we have \(\R' \red_{S} \S'\). By applying the isomorphism to the derivation we get \(\R \red_{S} \S\) as required. The second part of the lemma follows straight from the first, along with \cref{rem:using-webs} to restrict to non-trivial linear inferences and \cref{lem:minimality} to restrict to minimal inferences.%
\end{proof}

The above lemma allows us to search for independent linear inferences by only searching inferences from least webs to arbitrary webs. This increases the speed of the search greatly as it turns out there are relatively few isomorphism classes (and so least webs).
For example, specialising to the case \(n = 7\), there are \(78416\) \(P_4\)-free graphs with 7 variables with only \(180\) of them being least (so that there are 180 isomorphism classes).
Note that we may not similarly restrict the RHS of inferences to least webs. This means we need to know the maximal cliques of every \(P_4\)-free graph to determine whether there are inferences between them.

\smallskip

\textbf{Phase 3: generating all maximal cliques.}
In phase three we generate all the maximal cliques of the graphs found in phase one and store them so that they do not need to be recomputed every time we check a linear inference. Each maximal clique of a graph of size $n$ can be stored using only $n$ bits, so storing all this data is feasible.

\smallskip

\textbf{Phase 4: generating `least' linear inferences.}
With the maximal clique data, phase four of generating a list of all valid linear inferences (from a least web to an arbitrary web) can be easily done by iterating through all possible combinations and checking them using \cref{lem:rw-inference}.

\smallskip

\textbf{Phase 5: checking for non-triviality.}
Similarly phase five of checking which of these inferences are non-trivial is also simple using \cref{lem:rw-trivial}. This data is stored in a \texttt{HashMap} of sets for quick indexing.

\smallskip

\textbf{Phase 6: restricting to logically minimal inferences.}
Phase six is now to restrict our inferences to only those that are logically minimal. Write \(\Phi_\R \) be the set of webs distinct from $\R$ that \(\R \) (non-trivially) implies.
We calculate, for a least web $\R$, the set $M_\R$ of webs $\S$ with $\R \to \S$ a logically minimal linear inference using the identity:
\begin{equation*}
  M_\R = \Phi_\R \setminus \bigcup_{\R' \in \Phi_{\R}} \Phi_{\R'}
\end{equation*}
Note that to calculate this, we need to be able to generate \(\Phi_\R\) for arbitrary (i.e.\ not necessarily least) webs. This is where the isomorphism data stored in phase two becomes useful, as if \(\rho \) is an isomorphism from $\R$ to $\R'$, with \(\R'\) least, we can use,
\begin{equation*}
  \Phi_\R = \{ \rho^{-1}(\S) \ |\  \S \in \Phi_{\R'}\}
\end{equation*}
to generate \(\Phi_\R\), where we already have \(\Phi_{\R'}\). In the implementation, we generate each \(\Phi_\R\) on the fly (from \(\Phi_{\R'}\)), though we could have pre-generated all of these, which might provide further speedup for this phase.

\smallskip

\textbf{Phase 7: checking for \(S\) derivability.}
The last phase is to check the remaining inferences, of which there are now few enough to feasibly do so. Logically minimal inferences have one further benefit: a logically minimal inference (and in fact any \(S\)-minimal inference) \(\phi \to \psi\) is derivable from \(S\) if and only if it is derivable from a \emph{single} \(S\) step.

To check whether each remaining inference is derivable from \(S\), we use the library function for determining whether a given inference is an instance of another. For the specific case of \(S = \{\s,\m\}\), It would have been possible to use the criterion for switch derivability from \cite{Str07:char-med} (mentioned at the end of \cref{sec:webs}) and the criterion for medial derivability given in \Cref{prop:medial-char}, but running through possible partitions of the nodes of \(\R\) was fast enough and easier to implement.

\smallskip

\textbf{Evaluation and main results.}
The main results of the paper are obtained by running the above phases for \(n = 7\) and \(S = \{\s,\m\}\). We found that there were 78416 \(P_4\)-free graphs of which \(180\) were least. There were \(35110\) non-trivial inferences from a least web to an arbitrary web of which \(1352\) were minimal. Of these minimal inferences, \(968\) were an instance of switch, \(384\) were an instance of medial, and there were no other inferences.

\begin{proof}[Proof of \cref{thm:main-reduced}]
  By \cref{lem:least}, it suffices to show that all valid non-trivial logically minimal inferences from a least graph \(\mathcal{R}\) to any graph \(\mathcal{S}\) is derivable from switch and medial. Our algorithm computes the following.
  \begin{itemize}
  \item Phase 1 computes a set of \(n\)-variable \(P_4\)-free graphs \(X\).
  \item Phase 2 records the set \(X' = \{\mathcal{R} \in X | \mathcal{R}\text{ is least}\}\), which is computed by applying all possible permutations, checking if any produce a graph that is less than the original graph with respect to the numeric ordering on graphs.
  \item Phase 3 computes the maximal cliques of all graph in \(X\).
  \item Phase 4 uses the results of phase 3 to find all inferences from \(\mathcal{R} \in X'\) to \(\mathcal{S}\in X\). This is done by checking every combination of graphs against the criterion in \cref{lem:rw-inference}.
  \item Phase 5 computes \(\Phi_{\mathcal{R}}\) for each \(\mathcal{R} \in X'\), the set of webs non-trivially implied by \(\mathcal{R}\) that are distinct from \(\mathcal{R}\). The applies the criterion of \cref{lem:rw-trivial} against the results of the previous phase.
  \item Phase 6 filters the inferences of phase 5 to those that are logically minimal by removing any inferences that have an interpolating formula by exhaustive search.
  \item Phase 7 takes the inferences produced by phase 6 and removes any which are an instance of switch or medial.
  \end{itemize}
  For \(n = 7\), no inferences were produced by phase 7, meaning that all non-trivial logically minimal inferences from a least graph were derivable with switch and medial, which completes the proof.
\end{proof}

The algorithm can also run on 8 variables, again with \(S = \{\s,\m\}\), where there were 1320064 \(P_4\)-free graphs of which \(522\) were least. There were \(514486\) non-trivial inferences from a least web to an arbitrary web of which \(5364\) were minimal, Of these, \(3506\) were an instance of switch, \(1770\) were an instance of medial, and there were \(88\) other inferences. After quotienting out by isomorphism (as restricting to inferences from least graphs does not rule out self isomorphisms on the LHS of the inference), we were left with \(4\) inferences, of which two were dual to each other leaving the logically minimal $\{\s,\m\}$-independent inferences given in \cref{sec:8var-inf}. This completes the proof of our main result, \cref{thm:main}.

Furthermore, the algorithm is also fast enough to run on 9 variables, with \(S\) containing switch, medial, and the four 8-variable inferences given in \cref{sec:8var-inf}. In this case, there were 25637824 \(P_4\)-free graphs of which \(1532\) were least. There were \(8374668\) non-trivial inferences found of which \(20553\) were minimal. Of the remaining:
\begin{itemize}
    \item \(12333\) were an instance of switch,
    \item \(7212\) were an instance of medial,
    \item \(168\) were an instance of Inference \eqref{eq:php32-derived-inf-repeated},
    \item \(168\) were an instance of Inference \eqref{eq:unknown-inference-8-repeated},
    \item \(384\) were an instance of Inference \eqref{eq:counterexample-inference-repeated},
    \item \(104\) were an instance of the dual of Inference \eqref{eq:counterexample-inference-repeated},
    \item and there were \(184\) remaining inferences.
\end{itemize}
After quotienting by isomorphism, there were \(10\) remaining inferences, in \(5\) dual pairs. These inferences are given in \Cref{sec:9var-infs}, and a theoretical justification for this search is given in \cref{sec:further-theoretical}.
\smallskip

\textbf{Other features of the implementation.}
The implementation executable has a couple of extra features. One of these is the ability to run the search algorithm recursively on \(n\) to get a set of minimal inferences \(\basis(n)\). This works in the following way: We set \(\basis(0) = \emptyset\). To calculate \(\basis(n + 1)\), we first recursively run the algorithm on \(n\) to get a set of inferences \(\basis(n)\). We can then run the above phases on size \(n+1\) inferences checking against the set \(\basis(n)\) of inferences. This will give us a set \(S\) of \(n+1\) variable inferences. We then return \(\basis(n+1) = \basis(n) \cup S\). The use of this algorithm and it's name are explained in \Cref{sec:further-theoretical}.

Further we also include a flag that toggles whether the search algorithm searches for only \(P_4\)-free graphs, or if it searches for any graph. Turning off the checking for \(P_4\)-free graphs changes phase 1 of the algorithm to include all graphs instead of just those that are \(P_4\)-free, but the rest of the algorithm functions the same. We lastly also include a flag for using maximal stable set entailment (condition \eqref{item:valid-char-stable-sets}), instead of the maximal clique entailment (condition \eqref{item:valid-char-cliques}). Since both of these coincide for \(P_4\)-free graphs, this only needs to be considered when the ``\(P_4\)'' setting is turned off.

These extra features of the implementation are not needed for the proof of \cref{thm:main-reduced}, but are used in the following sections.

\section{Enumerating a `basis' of linear inferences}
\label{sec:further-theoretical}

\Cref{thm:main} was concerned with minimal linear inferences independent of switch and medial, and so its proof requires the special case of $S = \{\s,\m\}$ in the search algorithm of the previous section.
This section is devoted to proving that the recursive algorithm $\basis(n)$ described at the end of the last section in fact computes a `canonical' basis of linear inferences independent of all smaller ones.

The content of this section is new and does not appear in the preliminary version \cite{DasRice21}.

\subsection{Defining a `basis' of linear inferences}
\label{sec:basis}
First, we must make precise what the `canonical' linear inferences of a certain size are.

\begin{defi}
Let $L_{\leq n}$ denote the set of all linear inferences on $\leq n$ variables, and write $L_{<n} = L_{\leq(n-1)}$.
For $n>4$, we call a set $X$ of linear inferences \textbf{$n$-minimal} if it is a smallest set of linear inferences such that $X \cup L_{<n}$ derives (with units) every $\mathsf r \in L_n$.
\end{defi}

For instance, we already have from our previous results:

\begin{exa}
For $n=5,6,7$ we have $\emptyset$ is $n$-minimal by \Cref{thm:main}.
By the outputs of our search algorithm up to 8 variables, cf.~\Cref{sec:8var-inf}, we have that $ \{ \eqref{eq:php32-derived-inf-repeated}, \eqref{eq:unknown-inference-8-repeated}, \eqref{eq:counterexample-inference-repeated},\overline{\eqref{eq:counterexample-inference-repeated}}\} $ is $n$-minimal for $n = 8$.
\end{exa}

\begin{prop}
\label[prop]{lem:n-min-set-char}
Let $X$ be a $n$-minimal set, for some $n>4$. We have:
\begin{enumerate}
\item If $X \ni \mathsf r : \phi \to \psi$ then $\phi$ and $\psi$ contain the same variables.
\item  Any $\mathsf r \in X$ is non-trivial.
    \item Any $\mathsf r \in X$ is logically minimal.
    \item Any non-trivial logically minimal $\mathsf r \in L_{\leq n}$ that is not $L_{<n}$-derivable (with units) is an instance of some $\mathsf r' \in X$, modulo $\acu$.
\end{enumerate}
\end{prop}
\begin{proof}
For (1), suppose for contradiction that $\phi$ has variables $\V_1$ and $\psi$ has variables $\V_2$ where $\V_1 \neq \V_2$. However by \cref{prop:diff-vars}, we then have that there is some $\mathsf r'$ on $\V_1 \cap \V_2$, which must have fewer variables than $\mathsf r$, and that $\mathsf r$ is derivable from $\mathsf r'$, switch and medial, meaning $\mathsf r$ is already $L_{<n}$-derivable.

For (2), note that by \Cref{prop:trivial} any trivial $n$-variable linear inference is already $L_{<n}$-derivable.

For (3), suppose $X \ni \mathsf r: \phi \to \psi$ is not logically minimal.
Let $\chi$ have $n$ variables such that $\phi \to \chi\to\psi$ are valid but $\chi$ is not equivalent to $\phi$ or $\psi$.
Note in particular that $\phi,\chi,\psi$ are satisfied by $k,l,m$ assignments respectively, with $k<l<m$.
By assumption $X\cup L_{<n}$ must derive $\mathsf r_0: \phi \to \chi$ and $\mathsf r_1: \chi \to \psi$; we may assume that any such derivation does not use $\mathsf r$:
\begin{itemize}
    \item if an instance of $\mathsf r$ has a constant substituted for one of the variable inputs, then it reduces to a $L_{<n}$ inference, modulo $\acu$ by \Cref{prop:unit-free}.
    \item otherwise an instance of $\mathsf r$ would have $k$ assignments satisfying its LHS (which is too few in the case of $\mathsf r_1$) and $m$ assignments satisfying its RHS (which is too many in the case of $\mathsf r_0$).
\end{itemize}
Thus $(X\setminus \{\mathsf r\}) \cup L_{<n}$ derives $\mathsf r_0$ and $\mathsf r_1$ and hence it also derives $\mathsf r$ and by extension every $\mathsf r' \in L_{\leq n}$, contradicting the minimality of $X$.

For (4) let $L_{\leq n} \ni \mathsf r:\phi \to \psi$ be logically minimal and non-trivial and not $L_{<n}$-derivable.
$X \cup L_{<n}$ must derive $\mathsf r$, since it derives all of $L_{\leq n}$.
However, any such derivation cannot have any formula not already logically equivalent to $\phi $ or $\psi$, by logical minimality. Thus all formulae are $\acu$-equivalent to $\phi$ or $\psi$ by \Cref{acu-and-logical-equivalence}, and there must be a single step whose LHS and RHS are thus $\acu$-equivalent to $\phi $ and $\psi$ respectively.
\end{proof}

Note that the above lemma actually characterises $n$-minimal sets: a set $X$ is $n$-minimal iff it contains just (a renaming, modulo $\ac$, of) each non-trivial logically minimal $\mathsf r \in L_{\leq n}$ that is not $L_{<n}$-derivable.
\begin{cor}
For $n>4$, $n$-minimal sets are unique up to $\acu$ and renaming of variables.
\end{cor}

Due to \Cref{prop:const-free-neg-free}, and since every element of an $n$-minimal set is non-trivial, by \Cref{lem:n-min-set-char}, we can always choose the elements of an $n$-minimal set to be constant-free and negation-free. This motivates the following definition.

\begin{defi}
We define the sets $M_n$ for each $n$ as follows: For $n = 0,1,2$ we have $M_n = \emptyset$, then $M_3 = \{\s\}$, and $M_4 = \{\m\}$. For $n>4$, we write $M_n$ for the (unique up to $\ac$ and renaming) constant-free negation-free $n$-minimal set, and we write $M_{\leq n} := \bigcup_{m\leq n} M_m$.
\end{defi}

By definition of $n$-minimality and completeness of $\{\s,\m\}$ for $L_{\leq 4}$ we have:
\begin{fact}
\label[fact]{mn-compl-with-units}
For $n>4$, $M_{\leq n}$ is complete for $L_{\leq n}$, i.e.\ any $\mathsf r \in L_n$ is $M_{\leq n}$-derivable with units.
\end{fact}

In fact we can say more than this: $M_{\leq n}$ is the minimal set w.r.t.\ set inclusion (unique up to $\ac$ and renaming) such that, for every $k$, its restriction to $\leq k$-variable inferences derives every $\mathsf r \in L_{\leq k}$. It also has the property that for every $\mathsf r \in M_n$, we cannot derive $\mathsf r$ in $M_{\leq n} \setminus \{\mathsf r\}$. This can be viewed as a `local' version of independence that allows the unit-elimination results of the next subsection to go through.
In this way we may say that it forms a `basis' for $L$.

\begin{rem}
    [On independence]
    \label{rem:on-independence-of-Mn}
    Note that our sets $M_{\leq n}$ are not fully independent. For instance $\s$ is an $\acu$-instance of $\eqref{eq:php32-derived-inf-repeated}$.
In fact we believe that there cannot exist a minimal set (w.r.t.\ set inclusion) that is complete, i.e.\ derives all of $L$, as infinite descending sequences (w.r.t.\ set inclusion) may exist.
Indeed the inferences $\mathit{QHQ}_n$ from \cite{Str12:ext-wo-cut} seem to suggest this, since each $\mathit{QHQ}_n$ derives all $\mathit{QHQ}_{<n}$.
However these inferences are not logically minimal, and a comprehensive reduction of this family to a logically minimal one is beyond the scope of this work.

\end{rem}

\subsection{Computing \texorpdfstring{$M_n$}{Mn} recursively without units}
With a view to generalising our search algorithm, we would like to, for arbitrary $n$, identify $M_n$, i.e.\ the `minimal' $n$-variable inferences that are independent of $L_{<n}$.
However,
in order to apply our implementation, based on graph theoretic representations that do not account for constants or negation, we need the following generalisation of  \Cref{non-triv-const-free-neg-free-derivability-without-units}.

\begin{thm}
\label[thm]{thm:mn-withunits-to-withoutunits}
For $n\geq 4$,
if $M_{\leq n}$ derives a constant-free negation-free nontrivial linear inference $\mathsf r$, then it also derives $\mathsf r$ without units.
\end{thm}

\begin{proof}
We structure the whole proof by induction on $n$. The base case is given by \Cref{non-triv-const-free-neg-free-derivability-without-units}, and we follow a variation of the argument therein for the inductive step too.

Again, for any formula $\phi$, write $\phi'$ for the unique constant-free formula (or $\bot $ or $\top$) $\sim_\un$-equivalent to $\phi$, cf.~\cref{prop:unit-free}.
Suppose $\phi \to \psi$ is a nontrivial constant-free negation-free linear inference with an $M_{\leq n}$-derivation (with units) of the form:
\begin{equation}
    \label{eq:nontriv-mn-derivation}
    \phi \sim_\acu \phi_0 \to_{M_{\leq n}} \psi_0 \sim_\acu \cdots \sim_\acu \phi_n \to_{M_{\leq n}} \psi_n \sim_\acu \psi
\end{equation}
Like in \Cref{non-triv-const-free-neg-free-derivability-without-units}, we proceed by replacing each $\phi_i$ and $\psi_i$ by $\phi_i'$ and $\psi_i'$ respectively.
As before, we indeed have that $\psi_i' \sim_\ac \phi_{i+1}'$, thanks to \Cref{acu-and-logical-equivalence}.

For the inferences $\phi_i' \to \psi_i'$, we shall show by induction on the definition of $\to_{M_{\leq n}}$ (namely as the closure of $M_{\leq n}$ under contexts and substitution) that, whenever $\phi\to_{M_{\leq n}} \psi$, we have $\phi' \red_{M_{\leq n}} \psi'$.
(Note the difference here with \Cref{non-triv-const-free-neg-free-derivability-without-units} where we showed the stronger statement that $\phi' \sim_\ac \psi'$ or $\phi' \rightsquigarrow_\ms \psi'$, instead of $\phi' \red_\ms \psi'$).
\begin{itemize}
    \item In the case of context closure, the argument is the same as for \Cref{non-triv-const-free-neg-free-derivability-without-units}, relying on context closure of $\red_{X}$ in place of $\rightsquigarrow_X$.
    \item Suppose $\phi \to_{M_{\leq n}} \psi $ with $\phi = \alpha (\chi_0, \dots, \chi_k)$ and $\psi = \beta(\chi_0, \dots, \chi_k)$ where $\alpha (x_0, \dots, x_k) \to \beta(x_0, \dots, x_k)$ (all variables displayed) is in $M_{\leq n}$ (and so $k<n$).
    \begin{itemize}
        \item If each $\chi_i'$ is not a constant $\bot$ or $\top$ then, just as in \Cref{non-triv-const-free-neg-free-derivability-without-units}, we have $\phi' = \alpha (\chi_0', \dots, \chi_k') $ and $\psi' = \beta (\chi_0' , \dots, \chi_k')$ and so $\phi' \to \psi'$ is also an instance of $\alpha (x_0, \dots, x_n) \to \beta(x_0, \dots, x_n)$ in $M_{\leq n}$.
        \item Otherwise, $\phi' \to \psi'$ is an instance of some inference $\mathsf r: \gamma \to \delta$ on $<n$ variables, so that $\mathsf r \in L_{<n}$.
        Now we appeal to \cref{mn-compl-with-units} so that $\mathsf r$ is $M_{<n}$-derivable with units and so applying the main inductive hypothesis we have that $\mathsf r$ is $M_{<n}$-derivable \emph{without} units. Finally, since derivations are closed under substitution, we have that $\phi' \to \psi'$ is $M_{<n}$ derivable (without units), i.e.\ $\phi' \red_{M_{<n}} \psi'$ as required.
    \end{itemize}
\end{itemize}

Now, returning to \cref{eq:nontriv-mn-derivation}, we have by the above argument that each $\phi_i' \red_{M_{\leq n}} \psi_i'$, whence we conclude since $\red_{M_{\leq n}}$ is reflexive and transitive.
\end{proof}

Putting together \Cref{lem:n-min-set-char}, \Cref{mn-compl-with-units}, and \Cref{thm:mn-withunits-to-withoutunits} above we finally have:
\begin{cor}
    \label[cor]{cor:m-n-classification}
  For all $n$, $M_n$ is precisely the set containing (a renaming, modulo $\ac$, of) each logically minimal constant-free negation-free non-trivial $n$-variable linear inference that is not an instance without units of some inference in $M_{<n}$.
\end{cor}
\begin{proof}
    This is easily checked for $n \leq 4$. For $n > 4$, we have from \Cref{lem:n-min-set-char,mn-compl-with-units} that $M_n$ consists of (a renaming, modulo $\ac$, of) each logically minimal constant-free negation-free non-trivial $n$-variable linear inference that is not derivable by $M_{<n}$ with units.
    By \Cref{thm:mn-withunits-to-withoutunits}, we have that it is sufficient to be not derivable from $M_{<n}$ without units, and since inferences of $M_n$ are logically minimal, this is the same as each inference of $M_n$ not being an instance of any inference in $M_{<n}$.
\end{proof}

\subsection{Enumerating \texorpdfstring{$M_n$}{Mn} using graphs}
\label{sec:enum-m-n}

With the results of the previous section, we can now compute each $M_n$ using the techniques from \cref{sec:webs} and \cref{sec:algorithm}. We start with some definitions:

\begin{defi}
    Recall that graphs have an ordering $<$ derived from their numerical representation. Further recall that we call a graph least if it is the smallest element (with respect to this ordering) in its isomorphism class. Let an inference of graphs $\R \to \S$ be \textbf{least} if $\R$ is least and there is no isomorphism $\rho$ such that $\rho(\R) = \R$ and $\rho(\S) < \S$.

    Inductively define $W_n$ to be the set containing each logically minimal non-trivial $n$-variable least linear inference between $P_4$-free graphs which is not an instance of any inference in $W_{<n}$ (i.e., under substitutions of $P_4$-free graphs for nodes).
\end{defi}

We now claim the following:

\begin{lem}
$M_n$ contains just (a renaming of) a cograph decomposition of each $\mathsf r \in W_n$.
\end{lem}
\begin{proof}
The proof proceeds by induction on $n$ by using the classification of \Cref{cor:m-n-classification}.

First, all cograph decompositions of inferences in $W_n$ are logically minimal non-trivial $n$-variable inferences by definition, and must be constant-free and negation-free, as they are decompositions of graph-based inferences. Further they are not an instance (without units) of some inference in $M_{<n}$ as this would render them an instance of an inference in $W_{<n}$ by the inductive hypothesis.
Since two formulae have the same web if and only if they are $\sim_\ac$ equivalent, a renaming of some cograph decomposition of each inference in $W_n$ is present in $M_n$.

Conversely, suppose $\mathsf r : \phi \to \psi$ is a logically minimal constant-free negation-free non-trivial $n$-variable inference that is not an instance (without units) of some inference in $M_{<n}$. Then there is some isomorphism $\rho$ such that $\rho(\W(\phi)) \to \rho(\W(\psi))$ is a logically minimal non-trivial $n$-variable least linear inference. Further it is not an instance of $W_{<n}$ by inductive hypothesis and assumption, and so $\rho(\W(\phi)) \to \rho(\W(\psi)) $ is in $W_n$.
Note that $\mathsf r$ is, by definition, just a renaming of some cograph decomposition of $\rho(\W(\phi)) \to \rho(\W(\psi))$.
\end{proof}

We now find each $W_n$ by inductively finding $W_{\leq n}$ for each $n$. The base case is given that $W_{\leq 4} = \{\s,\m\}$.
For the inductive step we suppose we want to find $W_n$ and $W_{\leq n}$ for some $n > 5$, and assume that we already know $W_{<n}$. We simply use the search algorithm specified in \Cref{sec:search} with inputs \(n\) and using \(W_{<n}\) as the set of rewrites to check against. This will find us a set of \(n\)-dimensional inferences between \(P_4\)-free graphs which are logically minimal, valid, non-trivial, have a least graph as their premise, and are not derivable from \(W_{<n}\). From here, $M_n$ is calculated by simply taking cograph decompositions of the elements of $W_n$. This is exactly the procedure \(\basis\), described at the end of \Cref{sec:search}, which is part of the executable portion of the implementation.

\begin{thm}
The recursive algorithm \(\basis(n) \) introduced in \Cref{sec:search} computes the set \(M_n\).
\end{thm}

It should further be remarked that running this recursive algorithm for \(n = 7\) and \(n = 8\) replicates the main results of the paper, as \(M_{<7} = M_{<8} = \{\s,\m\}\), and so the recursive algorithm simply checks for switch medial derivability in these cases.

\section{New 9-variable minimal linear inferences}
\label{sec:9var-infs}
Running $\basis(9)$ found 10 distinct $9$-variable inferences which are logically minimal, non-trivial, and not derivable from switch, medial, and the four 8-variable inferences presented in \cref{sec:8var-inf}. These are the 5 inferences in \cref{fig:nine_vars}, along with each of their duals, which form the set $M_9$. None of the 9-variable inferences were self-dual.
Running $\basis(10)$ seemed infeasible on a standard desktop.

In what follows we provide some analysis of these new 9-variable inferences, though a more detailed inspection of these (and, indeed, larger inferences) could form the subject of future work.
The content of this section is new and does not appear in the preliminary version \cite{DasRice21}.

\begin{figure}
    \centering
    \begin{equation}
        \label{eq:nine_one}
        \begin{alignedat}{2}
        & &&((x \land x') \lor x'') \land ((y \land (y' \lor y'')) \lor (z \land (z' \lor z'')))\\
        &\to &\quad& ((z \lor y'') \land (x'' \lor (x' \land z'))) \lor ((x \lor y') \land (y \lor z''))
        \end{alignedat}
    \end{equation}
    \\
    \begin{equation}
        \label{eq:nine_two}
        \begin{alignedat}{2}
        & &&((x \land x') \lor x'') \land ((y \land (y' \lor y'')) \lor (z \land (z' \lor z'')))\\
        &\to &\quad& ((z \lor y'') \land (x'' \lor z')) \lor ((x' \lor y') \land (y \lor (x \land z'')))
        \end{alignedat}
    \end{equation}
    \\
    \begin{equation}
        \label{eq:nine_three}
        \begin{alignedat}{2}
        & &&((x \land x') \lor x'') \land ((y \land (y' \lor y'')) \lor (z \land (z' \lor z'')))\\
        &\to &\quad& (y \land (x \lor y')) \lor ((z \land (x' \lor z')) \lor (x'' \land (y'' \lor z'')))
        \end{alignedat}
    \end{equation}
    \\
    \begin{equation}
        \label{eq:nine_four}
        \begin{alignedat}{2}
        & &&((w \land (x \lor x')) \lor w') \land ((y \land y') \lor (z \land (z' \lor z'')))\\
        &\to &\quad& ((w' \lor (w \land (x \lor y))) \land (y' \lor z')) \lor (z \land (x' \lor z''))
        \end{alignedat}
    \end{equation}
    \\
    \begin{equation}
        \label{eq:nine_five}
        \begin{alignedat}{2}
        & &&((w \land (x \lor x')) \lor w') \land ((y \land y') \lor (z \land (z' \lor z'')))\\
        &\to &\quad& ((w' \lor (x' \land y)) \land (y' \lor z')) \lor ((x \lor z) \land (w \lor z''))
        \end{alignedat}
    \end{equation}
    \caption{9-variable inferences which form $M_9$ with their duals.}
    \label[fig]{fig:nine_vars}
\end{figure}

\setcounter{subsection}{-1}

\subsection{General comments}
Before turning to each inference individually, let us remark on some global patterns and statistics.

\subsubsection{Recurring formula structures}
Note that the inferences \eqref{eq:nine_one}, \eqref{eq:nine_two} and \eqref{eq:nine_three} have the same LHS, and also \eqref{eq:nine_four} and \eqref{eq:nine_five} have the same LHS.
On the other hand  \eqref{eq:nine_one}, \eqref{eq:nine_two} and \eqref{eq:nine_five} have the same RHS (up to renaming, associativity and commutativity).
The prevalence of certain formula structures for the LHS and RHS of minimal linear inferences feels somewhat surprising and is perhaps worthy of further investigation.

Of course the dual properties for the RHSs and LHSs of the dual inferences naturally hold too: the duals of \eqref{eq:nine_one}, \eqref{eq:nine_two} and \eqref{eq:nine_three} have the same RHS; the duals of \eqref{eq:nine_four} and \eqref{eq:nine_five} have the same RHS; and the duals of \eqref{eq:nine_one}, \eqref{eq:nine_two} and \eqref{eq:nine_five} have the same LHS (up to renaming, associativity and commutativity).

\subsubsection{Numbers of red/green edges and connectives.}
Each of the inferences in \Cref{fig:nine_vars} has more red edges (equivalently, fewer green edges) in the LHS than in the RHS.
Apart from medial, every minimal linear inference thus far encountered in the literature has this property, and it seems natural to expect this to be the case for all inferences in $M_{n}$ for $n>4$.

Each inference also has more $\land$s (equivalently, fewer $\lor$s) in the LHS than in the RHS.
In fact, each inference has precisely 4 $\land$s (and so 4 $\lor$s too) in the LHS, and 3 $\land$s (and so 5 $\lor$s) in the RHS.
Again, note that dual properties for the dual inferences hold too.

\subsubsection{Relevance to conjectures from \cite{DasStr16:no-compl-lin-sys}}
All the inferences except for \eqref{eq:nine_three} serve as further counterexamples to Conjecture 7.9 from \cite{DasStr16:no-compl-lin-sys}.
Conjecture 7.8 from that work concerned minimal linear inferences that preserve (or increase) the number of $\land$s. None of the new 8 or 9 variable inferences presented in this work have this property, and it is indeed possible that each inference in $M_{>3}$ always decreases the number of $\land$s.

All of the 8 variable inferences in $M_8$ were instances (under units, associativity and commutativity) of some inferences in $M_9$.
Switch is also an instance of each inference of $M_9$.
However, despite the now numerous counterexamples to Conjecture 7.9 from \cite{DasStr16:no-compl-lin-sys}, it is interesting to note that medial is not an instance of any of the 8 or 9 variable inferences in this work (cf.~also \cite[Conjecture 7.10]{DasStr16:no-compl-lin-sys}).

\subsection{Inference \eqref{eq:nine_one}}
The relation webs corresponding to inference \eqref{eq:nine_one} are given by:
\[ \NineGraph{x,x',x'',y,y',y'',z,z',z''}rgrrrrrrgrrrrrrrrrrrrrrggggggggggrrg
\quad \raisebox{30pt}{$\to$} \quad
\NineGraph{x,x',x'',y,y',y'',z,z',z''}ggrggggrgggrrrgggrrggrgggggggrgrgrgg \]
\begin{itemize}
    \item The LHS has 23 red edges (and 13 green edges), and the RHS has 11 red edges (and 25 green edges).
    \item The inference flips 14 red edges to green, and 2 green edges to red.
    \item The inference generalises switch, but also \eqref{eq:unknown-inference-8-repeated} (by the substitution $z'' \mapsto \bot$) and $\overline{\eqref{eq:counterexample-inference-repeated}}$ (by the substitution $x'\mapsto \top$).
\end{itemize}

\subsection{Inference \eqref{eq:nine_two}}
The relation webs corresponding to inference \eqref{eq:nine_two} are given by:
\[ \NineGraph{x,x',x'',y,y',y'',z,z',z''}rgrrrrrrgrrrrrrrrrrrrrrggggggggggrrg
\quad \raisebox{30pt}{$\to$} \quad
\NineGraph{x,x',x'',y,y',y'',z,z',z''}rggrgggrgrggggrggrrggrgggggggrgrgrgg \]
\begin{itemize}
    \item The LHS has 23 red edges (and 13 green edges), and the RHS has 11 red edges (and 25 green edges).
    \item The inference flips 14 red edges to green, and 2 green edges to red.
    \item The inference generalises switch, but also \eqref{eq:php32-derived-inf-repeated} (by the substitution $z' \mapsto \bot$) and $\overline{\eqref{eq:counterexample-inference-repeated}}$ (by the substitution $0\mapsto \top$).
\end{itemize}

Note that both this and the last inference, \eqref{eq:nine_one}, are quite similar.
Not only do they have the same LHS, their RHSs differ only in the positioning of $x$ and $x'$.
This similarity reminds us of the similarity between \eqref{eq:php32-derived-inf-repeated} and \eqref{eq:unknown-inference-8-repeated}, and indeed this is reflected by the way those 8-variable inferences are duly generalised by these 9-variable ones.

\subsection{Inference \eqref{eq:nine_three}}
The relation webs corresponding to inference \eqref{eq:nine_three} are given by:
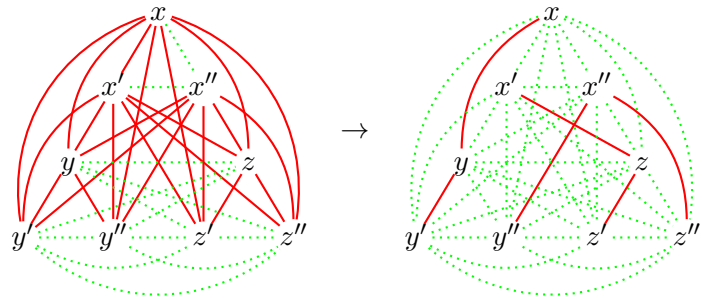
\[ \NineGraph{x,x',x'',y,y',y'',z,z',z''}rgrrrrrrgrrrrrrrrrrrrrrggggggggggrrg
\quad \raisebox{30pt}{$\to$} \quad
\NineGraph{x,x',x'',y,y',y'',z,z',z''}ggrgggggggggrggggrggrrgggggggggggrgg \]
\begin{itemize}
    \item The LHS has 23 red edges (and 13 green edges), and the RHS has 6 red edges (and 30 green edges).
    \item The inference flips 17 red edges to green, and 0 green edges to red.
    \item This inference generalises switch, but none of the 8-variable inferences from $M_8$.
\end{itemize}
Note that this is the only inference of $M_9$ that does not flip a green edge to a red one, and thus is the only inference not to constitute a counterexample to Conjecture 7.9 from \cite{DasStr16:no-compl-lin-sys}.
In this way, since only red edges are changed to green ones, we may say that it is a `switch-like' inference.\footnote{Note that \cite{DasStr16:no-compl-lin-sys} showed that any `medial-like' inference, defined similarly, is already medial-derivable.}
It would be interesting to investigate how sparse or dense such inferences, in fact, are.
Indeed, it is not even clear whether there are infinitely many.
The inferences $\mathit{QHQ}_n$ from \cite{Str12:ext-wo-cut} would serve as natural starting points for such considerations.

Let us also point out that the RHS has a quite regular form, being a disjunction of three formulas each of the form $a\land (b\lor c)$.

\subsection{Inference \eqref{eq:nine_four}}
The relation webs corresponding to inference \eqref{eq:nine_four} are given by:
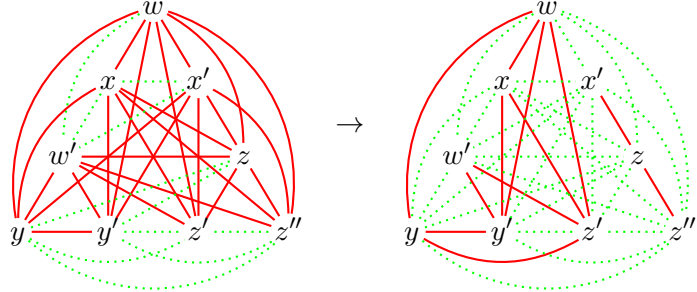
\[ \NineGraph{w,x,x',w',y,y',z,z',z''}rrgrrrrrggrrrrrgrrrrrrrrrrrggggggrrg
\quad \raisebox{30pt}{$\to$} \quad
\NineGraph{w,x,x',w',y,y',z,z',z''}rggrrgrggggrgrggggrgggrgrgrgrgggggrg \]
\begin{itemize}
    \item The LHS has 25 red edges (and 11 green edges), and the RHS has 12 red edges (and 24 green edges).
    \item The inference flips 14 red edges to green, and 1 green edge to red (namely the edge $(y,z')$).
    \item This inference generalises switch, but also \eqref{eq:unknown-inference-8-repeated} (by the substitution $x\mapsto \bot$).
\end{itemize}

Let us point out that this inference is distinguished in that its RHS has the highest logical depth among all inferences in $M_{\leq 9}$, exhibiting 4 alternations between $\lor $ and $\land$.
It would be interesting to see if inferences in $M_n$ become arbitrarily deep as $n$ increases.

\subsection{Inference \eqref{eq:nine_five}}
The relation webs corresponding to inference \eqref{eq:nine_five} are given by:
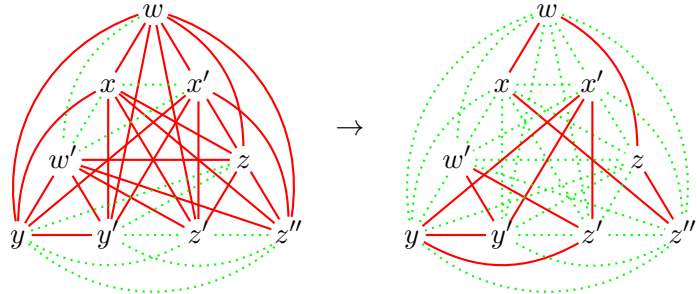
\[ \NineGraph{w,x,x',w',y,y',z,z',z''}rrgrrrrrggrrrrrgrrrrrrrrrrrggggggrrg
\quad \raisebox{30pt}{$\to$} \quad
\NineGraph{w,x,x',w',y,y',z,z',z''}rggggrggggggggrgrrgrggrgrgrgrgggggrg \]
\begin{itemize}
    \item The LHS has 25 red edges (and 11 green edges), and the RHS has 11 red edges and 25 green edges.
    \item The inference flips 15 red edges to green, and 1 green edge to red (namely the edge $(y,z')$).
    \item This inference generalises switch, but also the \eqref{eq:unknown-inference-8-repeated} (by the substitution $x\mapsto \bot$).
\end{itemize}

This inference is similar to the last one, in that the two have identical LHSs, each flips only one green edge to red, and each generalises the same inferences from $M_{\leq 8}$ (only \eqref{eq:unknown-inference-8-repeated}).
This inference is in a sense more `symmetric' in that the number of red edges in the LHS equals the number of green edges in the RHS (and vice-versa).

\section{Further remarks: applications to graph logics}
\label{sec:applications}

In this section we describe some further potential applications of our theoretical results in the previous section to \emph{graph logics}, and the generalised implementation available at \cite{Ric21:implementation}.

The content of this section is new and does not appear in the preliminary version \cite{DasRice21}.

\subsection{Graph inferences}
As mentioned earlier, our implementation is set up to be able to work with all graphs, and not just $P_4$-free graphs, in particular, we can check if an inference is an instance of any arbitrary graph rewrite. Using this, we can perform our algorithm without the restriction to $P_4$-free graphs to obtain sets $G_n$ of $n$-variable logically minimal non-trivial graph inferences which are not derivable (w.r.t.\ condition \eqref{item:valid-char-cliques}) from the set of smaller graph instances.
Note that although our software is able to check maximal stable set entailment (condition \eqref{item:valid-char-stable-sets}) as well as maximal clique entailment (condition \eqref{item:valid-char-cliques}), using the first would simply yield the dual of the results we list here.

There are still no non-trivial inferences for sizes $0$, $1$, and $2$. For $n = 3$, we have $G_3 = \{\s\}$. At $n = 4$ we get that $M_n \neq G_n$, as medial is no longer logically minimal, and in fact decomposes into the two following inferences:
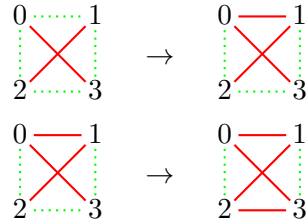
\[ \FourGraph{0,1,2,3}ggrrgg \quad \to \quad \FourGraph{0,1,2,3}rgrrgg \]
\[ \FourGraph{0,1,2,3}rgrrgg \quad \to \quad \FourGraph{0,1,2,3}rgrrgr \]
These two inferences form $G_4$.

We list $G_5$ in \Cref{fig:g-5} which has 16 elements. We were further able to calculate $G_6$ which has 137 elements, and $G_7$, which has 2013 elements, but we omit these here.

Let us point out, as a word of warning, that the sort of unit-elimination that we carry out for $\{\s,\m\}$-derivability in \cref{sec:preliminaries}, and more generally for linear derivability in \cref{sec:further-theoretical}, has not been considered.
The significance of our sets $G_n$ should be established by similar such results, but this is beyond the scope of this work.

\begin{figure}
    \centering
    \bgroup
    \setlength{\tabcolsep}{20pt}
    \begin{tabular}{cc}
    \(\FiveGraph{0,1,2,3,4}ggrrrggggg
    \quad \to \quad
    \FiveGraph{0,1,2,3,4}ggrrrrgggg\)&

    \(\FiveGraph{0,1,2,3,4}rgrrrggggg
    \quad \to \quad
    \FiveGraph{0,1,2,3,4}rgrrrggrgg\)\\

    \(\FiveGraph{0,1,2,3,4}grgrrrgggg
    \quad \to \quad
    \FiveGraph{0,1,2,3,4}gggrgrgggg\)&

    \(\FiveGraph{0,1,2,3,4}grgrrrgggg
    \quad \to \quad
    \FiveGraph{0,1,2,3,4}grrrrrgggg\)\\

    \(\FiveGraph{0,1,2,3,4}grgrrrgggg
    \quad \to \quad
    \FiveGraph{0,1,2,3,4}grgrrrgggr\)&

    \(\FiveGraph{0,1,2,3,4}rrgrrrgggg
    \quad \to \quad
    \FiveGraph{0,1,2,3,4}grgrgrgggg\)\\

    \(\FiveGraph{0,1,2,3,4}rrgrrrgggg
    \quad \to \quad
    \FiveGraph{0,1,2,3,4}rrgrrrgggr\)&

    \(\FiveGraph{0,1,2,3,4}grrrrrgggg
    \quad \to \quad
    \FiveGraph{0,1,2,3,4}grgrrggggg\)\\

    \(\FiveGraph{0,1,2,3,4}rrrrrrgggg
    \quad \to \quad
    \FiveGraph{0,1,2,3,4}rrgrrrgggg\)&

    \(\FiveGraph{0,1,2,3,4}rrgrrrgrgg
    \quad \to \quad
    \FiveGraph{0,1,2,3,4}rrgrrrgggg\)\\

    \(\FiveGraph{0,1,2,3,4}rrgrrrgrgg
    \quad \to \quad
    \FiveGraph{0,1,2,3,4}rggrrrgrgg\)&

    \(\FiveGraph{0,1,2,3,4}ggrrrgrrgg
    \quad \to \quad
    \FiveGraph{0,1,2,3,4}ggrgrggrgg\)\\

    \(\FiveGraph{0,1,2,3,4}rgrrrgrrgg
    \quad \to \quad
    \FiveGraph{0,1,2,3,4}ggrrrggrgg\)&

    \(\FiveGraph{0,1,2,3,4}rgrrrgrrgg
    \quad \to \quad
    \FiveGraph{0,1,2,3,4}rggrrgrggg\)\\

    \(\FiveGraph{0,1,2,3,4}rrrrrgrrgg
    \quad \to \quad
    \FiveGraph{0,1,2,3,4}rrrgrgrrgg\)&

    \(\FiveGraph{0,1,2,3,4}grrrrrrrgg
    \quad \to \quad
    \FiveGraph{0,1,2,3,4}grrrrgrrgg\)\\
    \end{tabular}
    \egroup

    \caption{The set of $G_5$ of logically minimal nontrivial graph inferences not derivable by all smaller ones.}
    \label[fig]{fig:g-5}
\end{figure}

\newcommand{\MLL}{\mathsf{MLL}}
\newcommand{\GS}{\mathit{GS}}
\subsection{The graph logic GS}
Acclavio, Horne and Stra{\ss}burger introduce in \cite{AccHorStr20:mll-graphs-short,AccHorStr20:mll-graphs-full} a graph logic motivated by the linear logic $\MLL$.
In particular they define an extension $\GS$ of $\MLL + \mathsf{mix}$ to arbitrary graphs, that uniquely satisfy certain desiderata.
Their logic has a primitive treatment of negation, inducing a form of graph entailment.
Interestingly, $\GS$ entailment is \emph{incomparable} with the conditions \eqref{item:valid-char-cliques} and \eqref{item:valid-char-stable-sets} on graphs considered in the previous subsection, even though the former's restriction to $P_4$-free graphs is contained in the latter's.

It would be interesting to adapt our implementation to consider $\GS$-entailment instead of the graph entailments of the previous subsection, in order to better understand theorems of $\GS$ itself, but this consideration is left for future work.

\section{Conclusions}
\label{sec:conclusions}
In this work we undertook a computational approach towards the classification of linear inferences.
To this end we succeeded in exhausting the linear inferences up to 8 variables, showing that there are four (distinct) 8 variable linear inferences that are independent of switch and medial, using the algorithm $\basis$ presented in \Cref{sec:algorithm}.
One of these new inferences (and its dual) contradicts a Conjecture 7.9 from \cite{DasStr16:no-compl-lin-sys}.
Conversely, all linear inferences on 7 variables or fewer are already derivable using switch and medial.

We proposed in \Cref{sec:enum-m-n} a well-defined notion $M_{\leq n}$ of `basis' for linear inferences with $\leq n$ variables.
We showed that our algorithm $\basis(n)$ in fact correctly computes each $M_n$.
We were further able to run $\basis(n)$ for $n=9$ and thus classified in \Cref{sec:9var-infs} the logically minimal 9 variable linear inferences independent of all inferences on $<9$ variables. Our preliminary analysis uncovered that all but 1 of these was a counterexample to the aforementioned conjecture, suggesting that inferences of this form may not be scarce when the number of variables is increased. We were not able to evaluate \(\basis(10)\), but believe it is feasible to compute with more resources or with more optimisations to the algorithm, which could be explored in further work. Little is known about the nature of the sets of inferences \(M_n\), or the patterns within these sets and we believe the tool presented in the paper could function as a starting point for conducting a more in depth analysis of these sets.

The set of linear inferences is equivalent to the multiplicative fragment of Japaridze's \emph{computability logic} (see, e.g., \cite{Jap03:intro-comp-log,comp-log-survey-webpage}), and also Blass' \emph{game semantics} for linear logic \cite{Bla92:game-semantics-ll}.
It would be interesting to examine our new linear inferences and their classifications in terms of the games therein, but that is beyond the scope of the present work.

Finally, our implementation is adaptable to a variety of logics and, in particular, graph-based systems such as those from \cite{AccHorStr20:mll-graphs-short,AccHorStr20:mll-graphs-full,CalDasWar20:bgl}.
To this end we gave some initial results and proof of concept in \Cref{sec:applications}.

\bibliographystyle{alphaurl}
\bibliography{citations}

\appendix

\section{Further proofs and examples}
\label{sect:app:further-proofs-examples}

\subsection{Validity of Equation \ref{eq:php32-derived-inf-repeated}}
\label{sec:app:validity-php32-derived}
We consider each assignment that satisfies the LHS and argue that it also satisfies the RHS:
\begin{itemize}
\item $\{z,x,x'\}$ satisfies $z\land (x \lor y)$.
\item $\{z,y,z'\}$ satisfies $z\land (x \lor y)$.
\item $\{z,y',z'\}$ satisfies $(w \lor y') \land ((w' \land x') \lor z')$.
\item $\{w, w', x, x'\}$ satisfies $(w \lor y') \land ((w' \land x') \lor z')$.
\item $\{w, w', y, z'\}$ and $\{w, w', y', z'\}$ satisfy $(w \lor y') \land ((w' \land x') \lor z')$.
\end{itemize}

\subsection{Validity of Equation \ref{eq:unknown-inference-8-repeated}}
\label{sec:app:validity-unknown-inf}
We consider each assignment that satisfies the LHS and argue that it also satisfies the RHS:
\begin{itemize}
    \item $\{z,x,x'\}$ satisfies $(z \lor (w \land x)) \land (x' \lor y)$.
\item $\{z,y,z'\}$ satisfies $(z \lor (w \land x)) \land (x' \lor y)$.
\item $\{z,y',z'\}$ satisfies $(w' \lor y') \land z'$.
\item $\{w, w', x, x'\}$ satisfies $(z \lor (w \land x)) \land (x' \lor y)$.
\item $\{w, w', y, z'\}$ and $\{w, w', y', z'\}$ satisfy $(w' \lor y') \land z'$.
\end{itemize}

\subsection{Validity of Equation \ref{eq:counterexample-inference-repeated}}
\label{sec:app:validity-counterexample-inference}
We consider each assignment that satisfies the LHS and argue that it also satisfies the RHS:
\begin{itemize}
\item $\{w,w',y,y' \}$ satisfies $w\land y$.
\item $\{w,w',z,z' \}$ satisfies $w' \land z' $ and $z$.
\item $\{ x,x',y,y' \}$ satisfies $x$ and $x' \land y' $.
\item $\{x,x',z,z' \}$ satisfies $x$ and $z$.
\end{itemize}

\subsection{\texorpdfstring{$\{\s,\m\}$}{\{s,m\}}-independence and \texorpdfstring{$\{\s,\m\}$}{\{s,m\}}-minimality of Equation \ref{eq:php32-derived-inf-repeated}}
\label{sec:ind-min-php32}

There are two possible medial applications to the subformula $(x \land x') \lor ((y \lor y') \land z')$ resulting in the following new LHSs:
\begin{itemize}
\item $(z \lor (w \land w')) \land (x \lor y \lor y') \land (x' \lor z')$. In this case $\{z, y', x'\}$ is a countermodel.
\item $(z \lor (w \land w')) \land (x \lor z') \land (x' \lor y \lor y')$. In this case $\{z, z', x'\}$ is a countermodel.
\end{itemize}

There are two possible switch applications to the subformula $(y \lor y') \land z'$ resulting in the following new LHSs:
\begin{itemize}
\item $(z \lor (w \land w')) \land ((x \land x') \lor y \lor (y' \land z'))$. In this case $\{w ,w',y\}$ is a countermodel.
\item $(z \lor (w \land w')) \land ((x \land x') \lor y' \lor (y \land z'))$. In this case $\{z, y'\}$ is a countermodel.
\end{itemize}

Finally any other switch application is on the top-level conjunction, resulting in a formula of the form $z \lor X$, $(w \land w') \lor X$, $(x \land x') \lor X$ or $((y \lor y') \land z') \lor X$, which admits a countermodel $\{z\}$, $\{w, w'\}$, $\{x, x'\}$ or $\{y ,z'\}$, respectively.

\subsection{\texorpdfstring{$\{\s,\m\}$}{\{s,m\}}-independence and \texorpdfstring{$\{\s,\m\}$}{\{s,m\}}-minimality of Equation \ref{eq:unknown-inference-8-repeated}}
\label{sec:ind-min-unknown}

There are two possible medial applications to the subformula $(x \land x') \lor ((y \lor y') \land z')$ resulting in the following new LHSs:
\begin{itemize}
\item $(z \lor (w \land w')) \land (x \lor y \lor y') \land (x' \lor z')$. In this case $\{z, z', x\}$ is a countermodel.
\item $(z \lor (w \land w')) \land (x \lor z') \land (x' \lor y \lor y')$. In this case $\{z, x, y'\}$ is a countermodel.
\end{itemize}

There are two possible switch applications to the subformula $(y \lor y') \land z'$ resulting in the following new LHSs:
\begin{itemize}
\item $(z \lor (w \land w')) \land ((x \land x') \lor y \lor (y' \land z'))$. In this case $\{w ,w',y\}$ is a countermodel.
\item $(z \lor (w \land w')) \land ((x \land x') \lor y' \lor (y \land z'))$. In this case $\{z, y'\}$ is a countermodel.
\end{itemize}

Finally any other switch application is on the top-level conjunction, resulting in a formula of the form $z \lor X$, $(w \land w') \lor X$, $(x \land x') \lor X$ or $((y \lor y') \land z') \lor X$, which admits a countermodel $\{z\}$, $\{w, w'\}$, $\{x, x'\}$ or $\{y ,z'\}$, respectively.

\subsection{\texorpdfstring{$\{\s,\m\}$}{\{s,m\}}-independence and \texorpdfstring{$\{\s,\m\}$}{\{s,m\}}-minimality of Equation \ref{eq:counterexample-inference-repeated}}
\label{sec:ind-min-counterexample-inference}
Let us first consider rules applicable to the LHS.
There are four possible medial applications, resulting in the following new LHSs:
\begin{itemize}
\item $(w \lor x) \land (w' \lor x' ) \land ((y \land y') \lor (z \land z'))$. In this case $\{w,x',y,y' \}$ is a countermodel.
\item $(w \lor x') \land (w' \lor x) \land ((y \land y') \lor (z \land z'))$. In this case $\{x',w',y, y' \}$ is a countermodel.
\item $((w \land w') \lor (x \land x')) \land (y \lor z) \land (y' \lor z' )$. In this case $\{ x,x',y,z' \}$ is a countermodel.
\item $((w \land w') \lor (x \land x')) \land (y \lor z' ) \land (y' \lor z)$. In this case $\{ w,w',z', y'  \}$ is a countermodel.
\end{itemize}
Any switch application to the LHS must be on the top-level conjunction, and will have the form $(a \land a') \lor X$, for $a \in \{w,x,y,z \}$. However, $\{w,w' \}$, $\{x,x' \}$, $\{y,y' \}$ and $\{z,z'\}$ are each countermodels for the RHS.

Now let us consider the possible rule applications leading to the RHS.
There are two possible medial instances, coming from the following new RHSs:
\begin{itemize}
\item $(w \land y) \lor (x \land x' \land y') \lor (w' \land z' \land z)$. In this case $\{x,x',z,z'\}$ is a countermodel.
\item $(w \land y) \lor (x \land z) \lor (w' \land z' \land x' \land y')$. In this case $\{ w,w',z,z' \}$ is a countermodel.
\end{itemize}
Now let us consider the switch instances:
\begin{itemize}
\item If the contractum of the switch is $x \lor (w' \land z')$, then $\{x,x',y,y'\}$ is a countermodel.
\item If the contractum of the switch is $(x' \land y') \lor z$, then $\{w,w',z,z'\}$ is a countermodel.
\item If the redex of the switch has the form $w \land X$ or $y \land X$, then $\{x,x',z,z'\}$ is a countermodel.
\item If the redex of the switch has the form $X \land (x \lor (w' \land z'))$ or $X \land ((x' \land y') \lor z)$, then $\{ w,w',y,y' \}$ is a countermodel.
\end{itemize}

\end{document}